\documentclass{article}

\usepackage[preprint]{neurips_2020}

\usepackage[utf8]{inputenc} 
\usepackage[T1]{fontenc}    
\usepackage{hyperref}       
\usepackage{url}            
\usepackage{booktabs}       
\usepackage{amsfonts}       
\usepackage{nicefrac}       
\usepackage{microtype}      
\usepackage{graphicx, subfigure}
\usepackage{bm}

\usepackage{algorithm, algcompatible}
\usepackage{amsmath, amssymb}
\usepackage{amsthm}
\usepackage{multirow}
\usepackage{complexity}
\usepackage{xcolor}
\usepackage{lineno}
\usepackage[normalem]{ulem}
\usepackage[title]{appendix}
\newtheorem{theorem}{Theorem}[section]
\newtheorem{lemma}[theorem]{Lemma}

\usepackage{bm}
\usepackage{tikz}
\usepackage{relsize}
\usepackage{algorithm, algcompatible}
\usepackage{graphicx, subfigure, booktabs, multirow}
\renewcommand{\v}[1]{\mathbf{#1}}

\newcommand{\aoned}{\texttt{1-diag}}
\newcommand{\anormd}{\texttt{$\ell_2$-norm-diag}}
\newcommand{\norm}[1]{\| #1 \|}
\usepackage{amsthm}
\newtheorem{conj}[theorem]{Conjecture}

\newtheorem{example}{Example}[section]

\newtheorem{defn}{Definition}[section]
\newtheorem{as}{Assumption}[section]

\newtheorem{claim}{Claim}
\newenvironment{subproof}[1][\proofname]{
  
  \renewcommand*{\proofname}{Proof of claim}
  \begin{proof}[#1]
}{
  \end{proof}
}

\title{Robust Correlation Clustering with Asymmetric Noise}

\author{%
  Jimit Majmudar  \\
  University of Waterloo\\
  \texttt{jmajmuda@uwaterloo.ca} \\
  \And
  Stephen Vavasis \\
  University of Waterloo \\
  \texttt{vavasis@uwaterloo.ca} \\
}

\begin{document}

\maketitle

\begin{abstract}
Graph clustering problems typically aim to partition the graph nodes such that two nodes belong to the same partition set if and only if they are similar. Correlation Clustering is a graph clustering formulation which: (1) takes as input a signed graph with edge weights representing a similarity/dissimilarity measure between the nodes, and (2) requires no prior estimate of the number of clusters in the input graph. However, the combinatorial optimization problem underlying Correlation Clustering is NP-hard. In this work, we propose a novel graph generative model, called the Node Factors Model (NFM), which is based on generating feature vectors/embeddings for the graph nodes. The graphs generated by the NFM contain asymmetric noise in the sense that there may exist pairs of nodes in the same cluster which are negatively correlated. We propose a novel Correlation Clustering algorithm, called \anormd, using techniques from semidefinite programming. Using a combination of theoretical and computational results, we demonstrate that \anormd\ recovers nodes with sufficiently strong cluster membership in graph instances generated by the NFM, thereby making progress towards establishing the provable robustness of our proposed algorithm.
\end{abstract}

\section{Introduction}
Suppose we have $n$ objects and for any two objects $i, j$, a similarity score $p_{ij} \in [0, 1]$, and we wish to determine a clustering of the objects such that objects in the same cluster are similar and objects in different clusters are dissimilar. Correlation Clustering, first introduced by \cite{bansal2004correlation}, formulates this problem as an optimization problem which does not require a priori knowledge about the number of clusters in the graph. The idea in Correlation Clustering is to first form a weighted graph on $n$ nodes where the weight of edge $ij$ is obtained from the similarity score using the transformation $\log(p_{ij}/(1-p_{ij}))$, and then to find a partition of the nodes which maximizes agreements, i.e. the sum of positive weights whose endpoints are put in the same cluster and the absolute values of negative weights whose endpoints are put in different clusters, (or, equivalently, minimizes disagreements, i.e. the sum of positive weights whose endpoints are put in different clusters and the absolute values of negative weights whose endpoints are put in the same cluster). In general, the aforementioned optimization problem is $\NP$-hard. Interestingly, the objective functions for disagreement minimization and agreement maximization differ by a constant, and as a result, an approximation algorithm provides different approximation ratio guarantees for the two problems; however, for the work presented in this work, the two problems are equivalent. While there has been considerable interest in designing approximation algorithms (\cite{bansal2004correlation, demaine2006correlation, makarychev2015correlation, mathieu2010correlation,  swamy2004correlation}), there have been very few works focusing on average case analysis or recovery of a ground truth clustering, as discussed in the following literature review. This style of analysis has recently gained popularity in tackling hard machine learning problems such as low-rank matrix completion (\cite{candes2009exact, recht2011simpler}), dictionary learning (\cite{arora2014more, spielman2012exact}), and overlapping community detection (\cite{anandkumar2013tensor, majmudar2020provable, mao2017mixed}), to name a few. In this work, we introduce a new graph generative model based on generating feature vectors/embeddings for the nodes in the graph, and propose a tuning-parameter-free semidefinite-programming (SDP)-based algorithm to recover nodes with sufficiently strong cluster membership. 

Among existing provable methods, \cite{joachims2005error} propose a fully-random model which generates signed graphs, and show that the model ground truth clustering is close to the optimal solution of the combinatorial optimization problem of maximizing agreements. However, it is not shown how to provably recover the model ground truth efficiently.

\cite{mathieu2010correlation} propose fully- and semi-random models, i.e. in which there is a probabilistic component and a deterministic adversarial component, which generate signed graphs. Their fully-random model can be interpreted as a special case of the planted partition model in which the noise probabilities $1-p$ and $q$ are equal, and the lack of an edge is treated as an edge with weight $-1$. For graph instances generated by the fully-random model, they propose a recovery algorithm which uses a modification of the SDP formulation proposed by \cite{swamy2004correlation} followed by a novel randomized rounding procedure. However, their proposed SDP formulation suffers from the limitation that it has $\Theta(n^3)$ (where $n$ is the number of nodes in the graph) constraints corresponding to triangle inequalities, making it almost unusable in practice for graphs with as low as $5000$ nodes. 

\cite{chen2014clustering} consider the planted partition model with the added difficulty that some entries of the adjacency matrix are not known. The input graphs are signed as the lack of an edge is treated as an edge with weight $-1$, and their algorithm uses a matrix-splitting SDP, originally introduced to express a given matrix as the sum of a sparse and a low-rank matrix. They provide conditions under which the SDP solution is integral and therefore their algorithm requires no rounding. They also argue that the recovered model ground truth clustering coincides with the optimal solution of the combinatorial optimization problem of minimizing disagreements.

\cite{makarychev2015correlation} propose a semi-random model which generates signed graphs. Their algorithm also uses the SDP proposed by \cite{swamy2004correlation} followed by a novel rounding procedure. Their recovery result states that if the input graph is generated from their semi-random model and additionally also satisfies some deterministic structural properties, then with constant probability at most a fraction of the nodes are mis-clustered.

\textbf{Notation}: For any natural number $n \in \mathbb{N}$, $[n]$ denotes the set $\{1, 2, \dots, n\}$, $\mathbb{R}^n$ denotes the vector space of $n$-dimensional real-valued vectors, and $\mathbb{S}^n$ denotes the vector space of $n \times n$ symmetric, real-valued matrices. Let $M$ be any matrix. We use $M_{ij}$ or $M(i, j)$ to denote its entry $ij$; for any set $\mathcal{R} \subseteq \mathbb{N}$, $M(\mathcal{R}, :)$ (resp. $M(:, \mathcal{R})$) denotes the submatrix of $M$ containing all columns (resp. rows) but only the rows (resp. columns) indexed by $\mathcal{R}$. For any two sets $\mathcal{R}, \mathcal{S} \subseteq \mathbb{N}$, $M(\mathcal{R}, \mathcal{S})$ denotes the submatrix of $M$ containing the rows indexed by $\mathcal{R}$ and the columns indexed by $\mathcal{S}$. We use $\max(M)$ to denote its largest value, and $M_+$ and $M_-$ to denote the projections onto the cone of non-negative and non-positive matrices respectively. To show that $M$ is a symmetric positive semidefinite (resp. positive definite) matrix, we use the notation $M \succeq 0$ (resp. $M \succ 0$). If $M$ is an $n \times n$ square matrix, $diag(M)$ denotes a vector in $\mathbb{R}^n$ whose entry $i$ is $M_{ii}$ for each $i \in [n]$, and $Diag(M)$ denotes an $n \times n$ matrix whose diagonal is equal $diag(M)$ and whose each off-diagonal entry is zero. Let $\v{v}$ be any vector. We use $v_i$ or $v(i)$ to denote its entry $i$; for any set $\mathcal{R} \subseteq \mathbb{N}$, $\v{v}(\mathcal{R})$ denotes the subvector of $\v{v}$ containing entries indexed by $\mathcal{R}$. We use $\max(\v{v})$ to denote its largest value, and $\v{v}_+$ and $\v{v}_-$ to denote the projections onto the cone of non-negative and non-positive vectors respectively. If $\v{v} \in \mathbb{R}^n$, $Diag(\v{v})$ denotes an $n \times n$ matrix whose diagonal is equal to $\v{v}$ and whose each off-diagonal entry is zero. If $\v{v}$ is an entry-wise non-negative vector, then for any $p \in \mathbb{R}$, we denote by $\v{v}^{\circ p}$ the vector obtained by exponentiating each entry of $\v{v}$ with $p$.

For any two matrices $X,Y$ of identical dimensions, $\langle X,Y \rangle$ denotes the trace inner product $trace(X^TY)$. For any two vectors $\v{u}, \v{v} \in \mathbb{R}^n$ such that for each $i \in [n]$, $u_i \leq v_i$, then we use $[\v{u}, \v{v}]$ to denote the set $\{\v{x} \in \mathbb{R}^n: u_i \leq x_i \leq v_i \ \forall i \in [n]\}$.

We use $\norm{\cdot}$ to denote the $\ell_2$-norm for vectors. The Frobenius norm of a matrix is denoted by $\norm{\cdot}_F$. $I$ and $E$ denote the identity matrix and the matrix with each entry set to one respectively whose dimensions will be clear from the context. For any positive integer $i$, $\v{e}_i$ denotes column $i$ of the identity matrix and $\v{e}$ denotes the vector with each entry set to one; the dimension of these vectors will be clear from context.

For any graph $G = (V, W)$, i.e. graph with node set $V$ and weighted adjacency matrix $W$, $L(G)$ denotes the graph Laplacian matrix defined to be $Diag(W\v{e}) - W$. For any subset $V' \subseteq V$ of nodes, $G[V']$ denotes the subgraph of $G$ induced by nodes in $V'$.

\section{Problem Formulation}

\subsection{Node Features Model (NFM)}\label{nfm}

We begin by defining the generative model, called the \emph{Node Features Model (NFM)}, for which we formulate the Correlation Clustering recovery problem. 

\begin{defn}[Node Features Model (NFM)]\label{nfm-def}
Let $n$ and $k$ be positive integers denoting the number of nodes and the number of clusters respectively. Let the nodes and the clusters be labelled using the sets $[n]$ and $[k]$ respectively. For each node $i \in [n]$, draw independently a feature vector $\bm{\theta}^i \in \mathbb{R}^k$ from a probability distribution on the unit simplex. Generate a weighted random graph $G$ on the $n$ nodes with weight matrix $W$ defined as
\begin{equation*}
    W_{ii'} = \begin{cases} \log\left( \dfrac{{\bm{\theta}^i}^T \bm{\theta}^{i'}}{1-{\bm{\theta}^i}^T \bm{\theta}^{i'}}\right) & \text{if $i\neq i'$} \\
    0 & \text{otherwise.}
    \end{cases}
\end{equation*}
For each $j \in [k]$, define cluster $V_j$ as
\begin{equation*}
    V_j := \{i \in [n]:\max(\bm{\theta}^i)>0.5, \arg \max(\bm{\theta}^i) = j \}
\end{equation*}
and define the set of stray nodes $V_{stray}$ as
\begin{equation*}
    V_{stray} := \{i \in [n]: \max(\bm{\theta}^i) \leq 0.5 \}.
\end{equation*}
\end{defn}

The intuition behind NFM is that first we generate a feature vector (or embedding) for each node in the graph, then for any pair $i, i' \in [n]$ of distinct nodes, we interpret ${\bm{\theta}^i}^T \bm{\theta}^{i'}$ as a similarity score, i.e. the probability with which the two nodes belong to the same cluster, lastly we apply a logarithmic transformation on the similarity score which produces a positive weight if the score is greater than $0.5$ and a negative weight if the score is less than $0.5$. The transformation $h(x) = \log(x/(1-x))$ which maps the set $(0, 1)$ to arbitrary real values is called the \emph{logit} or \emph{log-odds function} in literature, and its inverse $h^{-1}(x) = 1/(1+e^{-x})$ is the so-called \emph{logistic function}. These functions are commonly used in regression problems in which the output variable is interpreted as a probability and is therefore expected to belong to the set $(0, 1)$. For instance, a multivariate, vector-valued generalization of the logistic function, called the \emph{softmax function}, is widely used in classification problems to transform arbitrary real-valued vectors into probabilities corresponding to class memberships.  

We begin by asking the following question for the NFM described by (\ref{nfm-def}):
\begin{center}
    \textit{Given $W$, how can we efficiently recover the sets $V_1, \dots, V_k$ using no prior knowledge of $k$?}
\end{center}

Using a combination of theoretical analyses and computational experiments, we make progress towards answering the question posed above by proposing two SDP-based recovery algorithms, called \aoned\ and \anormd. The first recovery algorithm, \aoned, is studied in Section \ref{1d-sec} and is based on the SDP formulation of Swamy \cite{swamy2004correlation} whose variants have also been used in  \cite{makarychev2015correlation, mathieu2010correlation}. Then we demonstrate a limitation of the aforementioned algorithm to handle certain noisy instances. Consequently, we propose and analyze the novel \anormd\ recovery algorithm in Section \ref{nd-sec}. Our theoretical analysis is not comprehensive and the deficiencies are taken care of using evidence from computational experiments. Before proceeding to the material on the two recovery algorithms, in the subsequent sections, we discuss structural properties of the NFM relevant to the recovery problem we are interested in solving.  

\subsection{Nature of Noise in the NFM}\label{noise-nature}
We discuss the nature of noise in our model. Define the \emph{cluster set}
\begin{equation*}
    C_j := \{\v{x} \in \Delta^{k-1} : x_j > 0.5 \}
\end{equation*}
for each $j \in [k]$, and the \emph{central set}
\begin{equation*}
    C := \{\v{x} \in \Delta^{k-1} : \max(\v{x}) \leq 0.5 \}.
\end{equation*}
Figure \ref{sets} shows these sets for $k=3$.

\begin{figure}
    \centering
    \begin{tikzpicture}
\draw (0,0) node[anchor=north]{$(1, 0, 0)$}
  -- (4,0) node[anchor=north]{$(0, 1, 0)$}
  -- (2,3.46) node[anchor=south]{$(0, 0, 1)$}
  -- cycle;
\draw (2,0) node[anchor=north]{}
  -- (3,1.73) node[anchor=north]{}
  -- (1,1.73) node[anchor=south]{}
  -- cycle;
\draw (1, 0.7) node[] {$C_1$}
(3, 0.7) node[] {$C_2$}
(2, 2.4) node[] {$C_3$}
(2, 1.1) node[] {$C$};
\end{tikzpicture}
    \caption{Central set $C$ and cluster sets $C_1, C_2, C_3$ for the unit simplex in $\mathbb{R}^3$.}
    \label{sets}
\end{figure}
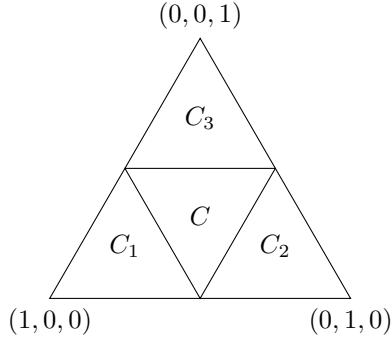

Note that in the light of the above definitions, we may equivalently redefine the sets $V_j$, for each $j \in [k]$, and $V_{stray}$ in Definition \ref{nfm-def} as
\begin{align*}
    V_j &:= \{i \in [n] : \bm{\theta}^i \in C_j \} \\
    V_{stray} &:= \{i \in [n] : \bm{\theta}^i \in C \}.
\end{align*}

Observe that for any $\v{x} \in C$ and $\v{y} \in \Delta^{k-1}$, $\v{x}^T\v{y} \leq 0.5$. This suggests that in the weighted graphs generated by the NFM, the stray nodes form negative edges with all other nodes in the graph, hence justifying their name. Due to this property, such nodes are quite benign with regards to mathematical analysis as any reasonable clustering algorithm, including the ones proposed in this work, ought to be able to detect them exactly. For any $\v{x} \in C_j$, $\v{y} \in C_{j'}$, for some distinct $j, j' \in [k]$, we have that $\v{x}^T\v{y} < 0.5$. This suggests that in the weighted graphs generated by the NFM, the clusters are \emph{well-separated} in the sense that each pair of nodes lying in distinct clusters shares a negative weight edge. However, if both $\v{x}, \v{y} \in C_j$, for some $j \in [k]$, then $\v{x}^T\v{y}$ may or may not be larger than $0.5$ and this is what introduces noise in our model. In other words, in the graphs generated by the NFM, it is possible for two nodes lying in the same cluster to share a negative weight edge. Therefore NFM models only one-sided noise. This behavior is well-motivated as real-world graphs do not always have a symmetric two-sided noise. For instance, consider a social network of researchers from the academic communities of mathematics, physics, history, and biology. Suppose the edge weights represent pair-wise similarities between any two researchers determined using the number of co-authored research articles. In this setting, we might have occasional collaborations amongst researchers of different communities; however, we almost certainly cannot expect all researchers in the same community to have collaborated with each other. In the language of weighted graphs, if the different academic communities represent the clusters in the graph, then we should expect a significantly high number of within-cluster negative edges compared to between-cluster positive edges. Due to such practical motivation, Correlation Clustering with asymmetric noise has also been studied in \cite{jafarov2020correlation, jafarov2021local}.

\subsection{Feature Space for a Cluster in the NFM}\label{feature-space}
As briefly mentioned in Section \ref{noise-nature}, it is possible for two nodes belonging in the same cluster to share a negative edge. It is instructive to understand further the nature of such negative edges. For each $j \in [k]$, define a partition of the set $C_j$ into \emph{strong} and \emph{fringe sets} as 
\begin{align*}
    C_j^{strong} &:= \{\v{x} \in \Delta^{k-1} : x_j \geq 1/\sqrt{2} \} \\
    C_j^{fringe} &:= \{\v{x} \in \Delta^{k-1} : 0.5 \leq x_j < 1/\sqrt{2} \}.
\end{align*}

Figure \ref{strong-fringe} shows these sets for $k=3$.

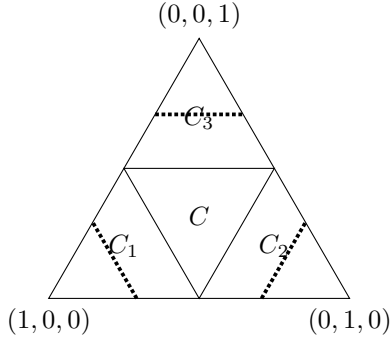
\begin{figure}
    \centering
    \begin{tikzpicture}
\draw (0,0) node[anchor=north]{$(1, 0, 0)$}
  -- (4,0) node[anchor=north]{$(0, 1, 0)$}
  -- (2,3.46) node[anchor=south]{$(0, 0, 1)$}
  -- cycle;
\draw (2,0) node[anchor=north]{}
  -- (3,1.73) node[anchor=north]{}
  -- (1,1.73) node[anchor=south]{}
  -- cycle;

\draw[densely dotted, line width = 0.5mm] (2.83,0) node[anchor=north]{}
  -- (3.415,1.012) node[anchor=north]{};
\draw[densely dotted, line width = 0.5mm] (1.17,0) node[anchor=north]{}
  -- (0.585,1.012) node[anchor=north]{};
\draw[densely dotted, line width = 0.5mm] (1.415,2.448) node[anchor=north]{}
  -- (2.585,2.448) node[anchor=north]{};

\draw (1, 0.7) node[] {$C_1$}
(3, 0.7) node[] {$C_2$}
(2, 2.4) node[] {$C_3$}
(2, 1.1) node[] {$C$};
\end{tikzpicture}
    \caption{Central set $C$ and the partition of corner sets $C_1, C_2, C_3$ into strong and fringe sets, shown using dotted lines, for the unit simplex in $\mathbb{R}^3$; for each corner set, the partition set containing a simplex vertex denotes the strong set.}
    \label{strong-fringe}
\end{figure}

Consequently, for each $j \in [k]$, we partition the cluster nodes $V_j$ into \emph{strong} and \emph{fringe nodes} as
\begin{align*}
    V_j^{strong} &:= \{i \in [n]: \bm{\theta}^i \in C_j^{strong}\} \\
    V_j^{fringe} &:= \{i \in [n]: \bm{\theta}^i \in C_j^{fringe}\}.
\end{align*}

These definitions are motivated by the intuition that the magnitude of the largest entry in the feature vector of a node quantifies the strength of cluster membership for that node. Moreover, the cut-off of $1/\sqrt{2}$ is chosen by noticing that any two points in the strong set of the same cluster have an inner product of at least $0.5$. In other words, any two nodes which are strong for the same cluster share a non-negative edge. Therefore if the graph contains only strong nodes for each cluster, then it has no noise in the form of a negative within-cluster edge. Fringe nodes, however, may potentially share some negative edges among themselves and with other nodes in the same cluster because the memberships of such nodes in their respective clusters are not sufficiently strong. Therefore we may think that it is difficult to cluster all the fringe nodes correctly.

\subsection{Relation to the MMSB}
The problem setup developed using the NFM bears some resemblance with the weighted version of MMSB considered in \cite{majmudar2020provable}. In particular, the graphs obtained by the NFM can be obtained by setting the community interaction matrix $B$ to be the identity matrix in the weighted MMSB. In terms of recovery, because we are modeling Correlation Clustering using the NFM, our goal is to recover only the cluster labels without using an a priori estimate of the number of clusters $k$. The weighted MMSB models the overlapping community detection problem in which the goal was to recover the fractional memberships of each node in the different communities and we were allowed to use a parameter corresponding to $k$ in the recovery algorithm.

\section{\aoned\ Recovery Algorithm}\label{1d-sec}
We first present and analyze the \aoned\ recovery algorithm which uses the SDP relaxation (\ref{p-1d}) first introduced in \cite{swamy2004correlation} to perform Correlation Clustering. For any node set $V$, we define the \emph{cluster matrix} for some partition of $V$ to be a $\lvert V \rvert \times \lvert V \rvert$, $0/1$ matrix whose entry $ii'$ is $1$ if and only if nodes $i$ and $i'$ belong to the same partition set.

\begin{algorithm}
\caption{\aoned}\label{alg-1d}
\textbf{Input:} Graph $G = (V, W)$ generated according to NFM \\
\textbf{Output:} Symmetric matrix $X^c$ of the same dimension as $W$ whose each entry is in $\{0, 1\}$
\begin{algorithmic}[1]
\STATE $X^* = \arg \max\ \langle W, X\rangle \text{ s.t. } X\geq 0, X\succeq 0, X_{ii} = 1 \; \forall i \in [n]$
\STATE $X^c = $ Round($X^*, 0.5$)
\IF{$X^c$ is not the cluster matrix for some partition of $V$}
\STATE $X^c = 0$
\ENDIF
\end{algorithmic}
\end{algorithm}

\begin{algorithm}
\caption{Round}\label{round}
\textbf{Input:} Matrix $X$, scalar $t$  \\
\textbf{Output:} $0/1$ matrix $X^r$ of the same dimension as $X$
\begin{algorithmic}[1]
\FOR{$i, j \in [n]$}
\STATE $X^r_{ij} = \begin{cases}
1 & \text{if } X_{ij} > t \\
0 & \text{if } X_{ij} \leq t
\end{cases}$
\ENDFOR
\end{algorithmic}
\end{algorithm}

Note that the output of \aoned\ can possibly be the zero matrix and therefore does not define a clustering for the input graph. However, the theory developed in Sections \ref{warmup} and \ref{1d-theory} provides conditions on the input graph sufficient for the output of \aoned\ to induce a clustering.  

\begin{equation}
\label{p-1d}
\tag{P-1D}
\begin{aligned}
& \max_X && \langle W, X \rangle \\
& \: \mathrm{s.t.} && X \geq 0 \\
&&& X \succeq 0 \\
&&& X_{ii} = 1 \ \forall i \in [n].
\end{aligned}
\end{equation}

\subsection{Warmup}\label{warmup}
We begin by analyzing the scenarios in which the the SDP (\ref{p-1d}) has a $0/1$ solution. The following theorem provides a deterministic sufficient condition on the graph instances for (\ref{p-1d}) to have a $0/1$ solution. Subsequently, we discuss the deterministic sufficient condition in the context of the NFM. 

\begin{theorem}\label{int-opt}
Let $G = (V, W)$ be a graph generated using the NFM. Suppose that for each $j \in [k]$, $L(G[V_j]) \succeq 0$. Let $X^*$ denote the cluster matrix corresponding to the partition $\{V_1, \dots, V_{k}, \{v\}_{v \in V_{stray}}\}$. Then $X^*$ is an optimal solution of (\ref{p-1d}).
\end{theorem}

Recall that $G[V_j]$ denotes the subgraph of $G$ induced by the node set $V_j$ and $L(G[V_j])$ denotes the Laplacian matrix of graph $G[V_j]$. The above theorem states that exact recovery of true clusters is achievable using (\ref{p-1d}) provided each cluster Laplacian is positive semidefinite. To connect this result with the NFM, we may quantify the probability such that each cluster Laplacian, in a graph instances generated by the NFM, is positive semidefinite. Table \ref{clust-lap-psd} shows some computational experiments in this regard. Each row in the table corresponds to $10$ cluster instances generated using the NFM in which the simplex distribution is chosen to be the Dirichlet distribution. We fix $k=3$ and the Dirichlet parameter $\bm{\alpha} = 0.3\v{e}$. The first column denotes the range in which the cluster size belongs, and the second column counts the \emph{PSD success}, i.e. number of cluster instances, out of $10$, which have a positive semidefinite Laplacian. Moreover, the third column contains the mean smallest eigenvalue of the Laplacian. 

\begin{table}
  \caption{Verification of positive semidefiniteness of cluster Laplacians.}
  \label{clust-lap-psd} 
  \centering
  \begin{tabular}{ccc}
    \toprule
    Cluster size range & PSD success $(/10)$ & Mean smallest Laplacian eigenvalue  \\
    \midrule
    $6-10$ & $4$ & $-1.31$ \\
    $11-15$ & $4$ & $-1.31$ \\
    $16-20$ & $0$ & $-3.60$ \\
    $21-25$ & $1$ & $-2.82$ \\
    $26-30$ & $1$ & $-4.26$ \\
    $31-35$ & $0$ & $-5.96$ \\
    $36-40$ & $0$ & $-6.35$ \\
    \bottomrule
  \end{tabular}
\end{table}

These computational results suggests a weakness of Theorem \ref{int-opt} in the sense that the determinstic condition required for exact recovery seems to hold with a probability converging to $0$ for the NFM with the Dirichlet distribution as the size of input graph grows. Morever, the decreasing smallest eigenvalue of the cluster Laplacians may also be interpreted as an increasing amount of noise in the clusters which motivates the following conjecture.

\begin{conj}\label{it-impos}
Let $G$ be a graph generated according to NFM in which the simplex distribution is chosen to be the Dirichlet distribution with constant parameter. No algorithm can exactly recover the true clusters in $G$ with probability not converging to $0$ as $n \to \infty$. 
\end{conj}

Conjecture \ref{it-impos} highlights information-theoretic limitations for exactly recovering the ground truth clusters, see \cite{banks2016information, banks2018information, berthet2020statistical, berthet2019exact, wang2016average} for instance, for results on information-theoretic limits for similar or related problems. The above observations also lead us to reformulate the central question posed in Section \ref{nfm} as follows. 

\begin{center}
    \textit{Given $W$, how can we efficiently recover exactly $k$ disjoint node sets, such that each node set contains exactly one of $V_1^{strong}, \dots, V_k^{strong}$, using no prior knowledge of $k$?}
\end{center}

We may interpret this reformulation as: instead of attempting to exactly recover the true clusters, we focus on exactly recovering the strong nodes, possibly in the presence of fringe nodes, which introduce noise in the form of negative within-cluster edges. The usage of the word ``contains'' in the above question indicates that recovery of any fringe node for a cluster is not necessarily intended but may happen. This perspective on robust Correlation Clustering, which involves clustering essentially only a subgraph of the input graph, is similar to that in \cite{krishnaswamy2019robust}, which provides an approximation algorithm for a generalized Correlation Clustering problem wherein the input graph is corrupted with a given number of noisy nodes which must be discarded before performing clustering. 

To answer the reformulated question above, we adopt a two-step algorithm analysis approach described as follows. Let $G$ be a graph generated by the NFM and let $\mathcal{A}$ be a cluster recovery algorithm of interest. For each $j \in [k]$, let $V_j'$ be the union of strong nodes and possibly some fringe nodes for cluster $j$, such that we expect $\mathcal{A}$ to successfully recover node sets $V_1', \dots, V_k'$ with some non-trivial probability as the number of nodes $n \to \infty$. In other words, $\mathcal{A}$ is likely to fail on the sets $V_j \setminus V_j'$ for each $j \in [k]$. We may formalize the behavior of $\mathcal{A}$ using the following two steps.

\begin{enumerate}
    \item For each $j \in [k]$, perturb the features of the node set $V_j \setminus V_j'$ to the central set to obtain the stray node set $V_j^{stray}$, and call the resulting graph $G'$. Prescribe deterministic conditions $\mathcal{C}'$ on node sets $V_j'$, for each $j \in [k]$, which ensure their exact recoverability from $G'$ by $\mathcal{A}$.
    \item For each $j \in [k]$, re-perturb the features of the node set $V_j^{stray}$ so as to obtain the node set $V_j \setminus V_j'$, which we may interpret as noisy nodes, i.e. we re-obtain graph $G$ from $G'$. Prescribe deterministic conditions $\mathcal{C}$ under which $\mathcal{A}$ is robust to the presence of node sets $V_j \setminus V_j'$, for each $j \in [k]$. The desired robustness properties are established by applying perturbation arguments to the analysis of $\mathcal{A}$ on $G'$ achieved in the previous step.
\end{enumerate}
In terms of probability quantification, we must also argue that for a graph $G$ generated by the NFM, the deterministic conditions required for provably robust recovery hold with probability not converging to $0$ as $n \to \infty$.

\subsection{Theoretical Guarantees}\label{1d-theory}
Using Theorem \ref{int-opt}, we conclude that if each cluster Laplacian is positive semidefinite, then \aoned\ achieves exact recovery. Adopting the two-step approach outlined in the previous section, we are now interested in the following two questions:
\begin{itemize}
    \item What is the probability that, for each cluster, the subgraph induced by the union of strong nodes and possibly some fringe nodes has a positive semidefinite Laplacian?
    \item Is the \aoned\ recovery algorithm robust to the presence of noisy nodes, i.e. fringe nodes that are close to being stray nodes?
\end{itemize}

In this section, we address the first question above, and in Section \ref{1d-robustness}, we address the second question. Observe that if we restrict our attention to the cluster subgraph induced by merely the strong nodes, then with probability $1$, the Laplacian is positive semidefinite because each edge has a non-negative weight. However, we are interested in extending this observation to a cluster subgraph induced by strong nodes and some fringe nodes which also possibly contains negative edges. (Based on the results in Table \ref{clust-lap-psd}, we cannot expect to include all fringe nodes.) For the NFM, directly quantifying the probability of Laplacian positive semidefiniteness for a cluster subgraph comprised of strong nodes and some fringe nodes appears a difficult task. Therefore in the following, Theorems \ref{psd1} and \ref{psd2} provide combinatorial sufficient conditions for a graph Laplacian to be positive semidefinite.

\begin{theorem}\label{psd1}
Let $G=(V, W)$ be a signed graph. Suppose for each negative edge $ii'$ where $i, i' \in [n]$, there exists a set of $m$ disjoint two-edge $ii'$-paths $\{P^{ii'}_l\}_{l \in [m]}$ of positive weights such that
\begin{equation*}
    -W_{ii'} \leq \sum\limits_{l \in [m]}\dfrac{1}{2} \times \text{harmonic mean of the two weights on $P^{ii'}_l$}
\end{equation*}
and the two-edge paths are disjoint across all negative edges, then $L(G) \succeq 0$. 
\end{theorem}
The intuition behind the proof of Theorem \ref{psd1} is to express the graph Laplacian as the sum of multiple graph Laplacians (corresponding to subgraphs of $G$), and then argue for the positive semidefiniteness of each summand Laplacian. Considering subgraphs in this way makes it easier to analyze negative edges; in particular a negative edge $ii'$ is included in a subgraph which also contains an adequate number of positive $ii'$-paths so as to compensate the contribution of the edge $ii'$ to the Laplacian. This idea is inspired by the support-graph technique used to design preconditioners for conjugate gradient \cite{bern2006support}. 

Theorem \ref{psd1} provides a sufficient condition to ensure Laplacian positive semidefiniteness, however, it is seemingly weak as described by the following example.

\begin{example}\label{2ep-weakness}
Generate a graph on $n$ nodes using the NFM with the probability distribution over the unit simplex fixed as the Dirichlet distribution. Suppose cluster $j$ of the graph contains $f_j$ fringe nodes. Consider a case in which a constant fraction of all pairs of the $f_j$ fringe nodes share a negative edge each. Then to use the sufficient condition in Theorem \ref{psd1} to ensure positive semidefiniteness of the Laplacian of cluster $j$, we require $\Omega(f_j^2)$ strong nodes for that cluster. In other words, if the cluster contains $n_j$ nodes, then Theorem \ref{psd1} allows for only $\mathcal{O}(\sqrt{n_j})$ fringe nodes. However letting $p$ be the probability of a feature vector lying in the fringe set for cluster $j$, we note that $\mathbb{E}[f_j] = np$. Moreover, using Hoeffding's inequality, we have that $f_j \in [np/2, 3np/2]$ with probability at least $1-2\exp(-np^2/2)$. That is, $f_j = \Theta(n)$, and consequently $f_j = \Omega(n_j)$, with probability converging to $1$ as $n \to \infty$. This suggests a potential weakness of the sufficient condition presented in Theorem \ref{psd1} for establishing positive semidefiniteness of cluster Laplacians.
\end{example}

The above shortcoming is addressed in the following theorem which provides a different combinatorial condition to ensure Laplacian positive semidefiniteness. 

\begin{theorem}\label{psd2}
Let $G=(V, W)$ be a signed graph. Let $U\subseteq V$ contain all nodes of $G$ adjacent to a negative edge. That is, $U := \{v \in V: W_{vw} < 0 \text{ for some } w \in V \}$. If there exists $S \subseteq V \setminus U$ such that for each $u \in U$ and $s \in S$, we have
\begin{equation}\label{psd2-as}
    \lvert S\rvert W_{us} \geq -2 \left( \sum\limits_{\substack{u' \in U: \\ W_{uu'} < 0}} W_{u u'} \right)
\end{equation}
then $L(G) \succeq 0$.
\end{theorem}

We revisit Example \ref{2ep-weakness} in the light of Theorem \ref{psd2}. If we assume that all edges in cluster $j$ other than the ones among the $f_j$ fringe nodes have a non-negative weight, and that the positive and negative weight magnitudes are of the same order, then to ensure positive semidefiniteness of the Laplacian of cluster $j$ using the sufficient condition obtained in Theorem \ref{psd2}, it suffices to have $f_j = \Theta(n_j)$. However, this example should not be interpreted to imply that Theorem \ref{psd2} is a strengthening of Theorem \ref{psd1}. For example, if we have a cluster in which each node is adjacent to a negative edge, Theorem \ref{psd1} may still be used to ensure positive semidefiniteness of the cluster Laplacian, but Theorem \ref{psd2} does not apply due to the absence of a set $S$. But for the purpose of analyzing a generative model such as the NFM, Theorem \ref{psd2} appears to be a better tool because of its tolerance to a number of fringe nodes that is linear in the size of the cluster, and because of the presence of strong nodes in the NFM. This is further corroborated by computational results shown in Table \ref{neg-strong-cond}. Each row in the table corresponds to $10$ cluster instances generated using the NFM in which the simplex distribution is chosen to be the Dirichlet distribution. We fix $k=3$ and the Dirichlet parameter $\bm{\alpha} = 0.3\v{e}$. The first column denotes the range corresponding to the size of the subgraph induced by strong nodes and fringe nodes whose feature vectors have largest entry at least $0.6$; the cut-off of $0.6$ is based on manual parameter search for the given setting of $k$ and $\bm{\alpha}$. The second column counts the \emph{combinatorial condition success}, i.e. number of instances, out of $10$, for which the subgraph satisfies (\ref{psd2-as}). 

\begin{table}
  \caption{Verification of sufficient condition (\ref{psd2-as}) for Laplacian positive semidefiniteness.}
  \label{neg-strong-cond} 
  \centering
  \begin{tabular}{ccc}
    \toprule
    Cluster size range & Combinatorial condition success $(/10)$ \\
    \midrule
    $6-10$ & $9$ \\
    $11-15$ & $9$ \\
    $16-20$ & $7$ \\
    $21-25$ & $9$ \\
    $26-30$ & $9$ \\
    $31-35$ & $8$ \\
    $36-40$ & $9$ \\
    \bottomrule
  \end{tabular}
\end{table}

These computational results suggest that the probability with which the cluster subgraphs consisting of nodes whose feature vectors have largest entry at least $0.6$ have a positive semidefinite Laplacian does not apparently converge to $0$ as $n \to \infty$, and also motivate the following conjecture. 

\begin{conj}
Let $G=(V,W)$ be a graph generated using the NFM in which the simplex distribution is chosen to be the Dirichlet distribution with constant parameter $\bm{\alpha}$. Then there exists a scalar $t(k, \bm{\alpha}) \in (0.5, 1/\sqrt{2})$ such that for each $j \in [k]$, with probability not converging to $0$ as $n \to \infty$, $G[V_j']$ satisfies the hypothesis of Theorem \ref{psd2} where
\begin{equation*}
    V_j' := \{i \in [n]: \theta^i_j \geq t(k, \bm{\alpha}) \}.
\end{equation*}
\end{conj}

\subsection{Proofs}
In this section, we include proofs of Theorems \ref{int-opt}, \ref{psd1}, and \ref{psd2} stated in Section \ref{1d-theory}.
\begin{proof}[Proof of Theorem \ref{int-opt}]
Our analysis uses SDP duality and therefore note that the dual of (\ref{p-1d}) is
\begin{equation}
\label{d-1d}
\tag{D-1D}
\begin{aligned}
& \min_{(Y, Z, \v{y})} && \v{e}^T \v{y} \\
& \: \mathrm{s.t.} && Y \geq 0 \\
&&& Z \succeq 0 \\
&&& W+Y+Z = Diag(\v{y}).
\end{aligned}
\end{equation}
As mentioned in the theorem statement, $X^*$ is the cluster matrix corresponding to the partition $\{V_1, \dots, V_{k}, \{v\}_{v \in V_{stray}}\}$.

Both optimization problems (\ref{p-1d}) and (\ref{d-1d}) have strictly feasible solutions. For instance, $X' := 0.5I + 0.5E$ is a positive, positive definite matrix which is feasible for (\ref{p-1d}). Similarly, $Y' := E$, $Z' := (\norm{W+E} + \epsilon) I - (W+E) $ and $\v{y}' := (\norm{W+E} + \epsilon) \v{e}$ gives a strictly feasible solution $(Y', Z', \v{y}')$ for (\ref{d-1d}) for any $\epsilon > 0$. Therefore using the Karush-Kuhn-Tucker (KKT) optimality conditions, we observe that $X^* \in \mathbb{S}^n$ is an optimal solution for (\ref{p-1d}) if and only if $X^*$ is feasible for (\ref{p-1d}) and there exists a feasible solution of (\ref{d-1d}), $(Y^*, Z^*, \v{y}^*)$ such that:
\begin{itemize}
    \item $X^*_{ij} Y^*_{ij} = 0, \forall i, j \in [n]$
    \item $\langle X^*, Z^*\rangle = 0$.
\end{itemize}
$X^*$ has non-negative entries with each diagonal entry being equal to one. Additionally, up to a permutation of its rows and columns, it is a block diagonal matrix in which each non-zero diagonal block is the matrix of all ones. Therefore $X^*$ is feasible for (\ref{p-1d}), and in the rest of the proof, we explicitly construct $(Y^*, Z^*, \v{y}^*)$.

For each $j \in [k]$, we set
    \begin{align*}
        Y^*(V_j, V_j) &= 0 \\
        Z^*(V_j, V_j) &= L(G[V_j]) \\
        \v{y}^*(V_j) &= W(V_j, V_j)\v{e}.
    \end{align*}
    
For each distinct $j, j' \in [k]$, we set
\begin{align*}
    Y^*(V_j, V_{j'}) &= -W(V_j, V_{j'}) \\
    Z^*(V_j, V_{j'}) &= 0.
\end{align*}

For each stray node $v \in V_{stray}$, we set
\begin{align*}
Y^*(v, :) &= -W(v, :) && \text{(and  $Y^*(:, v) = -W(:, v)$)} \\
Z^*(v, :) &= 0 && \text{(and  $Z^*(:, v) = 0$)} \\
y^*(v) &= 0.
\end{align*}

Because each pair of nodes lying in distinct clusters shares a negative edge and because each stray node shares a negative edge with every other node in the graph, we have that $Y^* \geq 0$. Similarly, because $L(G[V_j]) \succeq 0$ for each $j \in [k]$, we have that $Z^* \succeq 0$. 

Matrices $X^*$ and $Y^*$ have disjoint supports by construction, and therefore $X^*_{ij}Y^*_{ij} = 0$ for each $i, j\in [n]$. Moreover
\begin{align*}
    \langle X^*, Z^* \rangle &= \sum\limits_{j \in [k]}\langle X^*(V_j, V_j), Z^*(V_j, V_j) \rangle \\
    &= \sum\limits_{j \in [k]}\langle L(G[V_j]), E \rangle \\
    &= 0
\end{align*}
where the last line uses the fact that each row of a Laplacian matrix sums to zero.   

Lastly, we show that the equation $W + Y^* + Z^* = Diag(\v{y}^*)$ is satisfied. For each $j \in [k]$, we have
\begin{align*}
    W(V_j, V_j) + Y^*(V_j, V_j) + Z^*(V_j, V_j) &= W(V_j, V_j) + L(G[V_j]) \\
    & \hspace{2cm} (\text{using the definitions of $Y^*, Z^*$})\\
    &= Diag(\v{y}^*(V_j)). \\
    & \hspace{2cm} (\text{using the definition of $\v{y}^*$})
\end{align*}

For each distinct $j, j' \in [k]$, we have
\begin{align*}
    W(V_j, V_{j'}) + Y^*(V_j, V_{j'}) + Z^*(V_j, V_{j'}) &= 0
\end{align*}
using the definitions of $Y^*, Z^*$. Similarly, for each stray node $v$, we have
\begin{align*}
    W(v, :) + Y^*(v, :) + Z^*(v, :) &= 0 \\
    W(:, v) + Y^*(:, v) + Z^*(:, v) &= 0 
\end{align*}
using the definitions of $Y^*, Z^*$. 
\end{proof}

Now we provide proofs of Theorems \ref{psd1} and \ref{psd2} which provide combinatorial sufficient conditions for Laplacian positive semidefiniteness. 

\begin{proof}[Proof of Theorem \ref{psd1}]
Pick any negative edge $ii'$ in $G$, and let $i-i_1-i', \dots, i-i_m-i'$ denote $m$ disjoint two-edge $ii'$-paths of positive weights. Consider the subgraph of $G$ containing edge $ii'$ and these $m$ disjoint paths as shown in Figure \ref{2ep-subgraph}.

\begin{figure}
\centering
\begin{tikzpicture}
\draw (0,0) node[anchor=north]{$i$}
  -- (4,0) node[anchor=north]{$i'$};

\draw (0,0) node[anchor=north]{}
  -- (-0.5, 1) node[anchor=south]{$i_1$};
\draw (-0.5, 1) node[anchor=north]{}
  -- (4,0) node[anchor=north]{};

\draw (0,0) node[anchor=north]{}
  -- (0.5, 1) node[anchor=south]{$i_2$};
\draw (0.5, 1) node[anchor=north]{}
  -- (4,0) node[anchor=north]{};

\draw[dotted] (1.5, 1) node[anchor=north]{}
  -- (3.5, 1) node[anchor=north]{};

\draw (0,0) node[anchor=north]{}
  -- (4.5, 1) node[anchor=south]{$i_m$};
\draw (4.5, 1) node[anchor=north]{}
  -- (4,0) node[anchor=north]{};
\end{tikzpicture}
    \caption{Subgraph of $G$ containing negative edge $ii'$ and $m$ dijsoint two-edge $ii'$-paths of positive weights.}
\label{2ep-subgraph}
\end{figure}
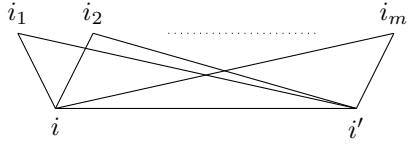

The contribution of this subgraph to the Laplacian of $G$ is the matrix, padded appropriately with zeros,
\begin{equation*}
    \begin{bmatrix}
    W_{ii'} + \sum\limits_{l \in [m]}W_{ii_l} & -W_{ii'} & -W_{ii_1} & -W_{ii_2} & \dots & -W_{ii_m} \\
    -W_{ii'} & W_{ii'} + \sum\limits_{l \in [m]}W_{i'i_l} & -W_{i'i_1} & -W_{i'i_2} & \dots & -W_{i'i_m} \\
    -W_{ii_1} & -W_{i'i_1} & W_{ii_1} + W_{i'i_1} & 0 & \dots & 0 \\
    -W_{ii_2} & -W_{i'i_2} & 0 & W_{ii_2} + W_{i'i_2} & \dots & 0 \\
    \vdots & \vdots & \vdots  & \vdots & \ddots & \vdots \\
    -W_{ii_m} & -W_{i'i_m} & 0 & 0 & \dots & W_{ii_m} + W_{i'i_m} 
    \end{bmatrix}.
\end{equation*}
Now since each of $W_{ii_1} + W_{i'i_1}, \dots, W_{ii_m} + W_{i'i_m}$ is positive, using the Schur complement condition for positive semidefiniteness, the above matrix is positive semidefinite if and only if the $2\times 2$ matrix
\begin{equation}\label{sc-2ep}
\begin{bmatrix}
W_{ii'} + \sum\limits_{l \in [m]}W_{ii_l} & -W_{ii'} \\
-W_{ii'} & W_{ii'} + \sum\limits_{l \in [m]}W_{i'i_l} 
\end{bmatrix} - \mathlarger{\sum}\limits_{l \in [m]} \left(
\dfrac{\begin{bmatrix}
W_{ii_l} \\ W_{i'i_l}
\end{bmatrix}\begin{bmatrix}
W_{ii_l} & W_{i'i_l}
\end{bmatrix}}{W_{ii_l} + W_{i'i_l}} \right)
\end{equation}
is positive semidefinite. However, the matrix in (\ref{sc-2ep}) can be rewritten as
\begin{equation*}
    \left(W_{ii'} + \sum\limits_{l \in [m]}\dfrac{W_{ii_l}W_{i'i_l}}{W_{ii_l} + W_{i'i_l}} \right) \begin{bmatrix}
    1 & -1 \\
    -1 & 1
    \end{bmatrix}
\end{equation*}
which is positive semidefinite if and only if 
\begin{equation*}
    -W_{ii'} \leq \sum\limits_{l \in [m]}\dfrac{W_{ii_l}W_{i'i_l}}{W_{ii_l} + W_{i'i_l}}
\end{equation*}
which proves the desired statement. 
\end{proof}

\begin{proof}[Proof of Theorem \ref{psd2}]
For notational ease, define $L:=L(G)$. Label the nodes of $G$ using the set $[n]$ and assume, without loss of generality, that $U=[m]$ for some $m < n$, and $S = \{m+1, \dots, m+\lvert S\rvert\}$. We define matrix $C \in \mathbb{S}^n$ as follows. For each $u, u' \in U$,
\begin{equation*}
    C_{uu'} := \begin{cases} L_{uu'} & \text{ if } u\neq u' \\
    \sum\limits_{l \in [m]\setminus \{u\}} \lvert L_{ul} \rvert & \text{ if } u=u'
    \end{cases}
\end{equation*}
Moreover $C(U, S) := \dfrac{-C(U, U)\v{e}\v{e}^T}{\lvert S \rvert}$ and $C(S, U) := C(U, S)^T$. Lastly, we set $C(S, S) := \dfrac{\v{e}^TC(U, U)\v{e}}{\lvert S\rvert} I$, and we set all other entries of $C$ to be zeros. That is,
\begin{equation*}
    C = \begin{bmatrix}C(U,U) & C(U, S) & 0 \\
    C(S,U) & C(S, S) & 0 \\
    0 & 0 & 0\end{bmatrix}.
\end{equation*}
In the rest of the proof, we argue that each of $C$ and $L-C$ is positive semidefinite, thereby proving the positive semidefiniteness of $L$. 

To show the positive semidefiniteness of $C$, it suffices to show the positive semidefiniteness of $C(U \cup S, U \cup S)$. First note that using the diagonal dominance property in $C(U,U)$, we conclude that $C(U,U)$ is positive semidefinite. Morever, since each node in $U$ is adjacent to at least one edge with a negative weight each entry of $C(U,U)\v{e}$ is positive. This implies that $\v{e}^TC(U,U)\v{e}$ is positive which in turn implies that $C(S, S)$ is invertible. Using the Schur complement condition for positive semidefiniteness, $C$ is positive semidefinite if and only if $C(U,U)-\dfrac{C(U,S)C(S,U)\lvert S\rvert}{\v{e}^TC(U,U)\v{e}}$ is positive semidefinite. Substituting for $C(U,S)$ and $C(S,U)$, we get
\begin{equation*}
    \begin{aligned}
    C(U,U)-\dfrac{C(U,S)C(S,U)\lvert S\rvert}{\v{e}^TC(U,U)\v{e}} &= C(U,U)-\dfrac{C(U,U)\v{e}\v{e}^T\v{e}\v{e}^TC(U,U)}{\lvert S\rvert\v{e}^TC(U,U)\v{e}} \\
    &= C(U,U)-\dfrac{C(U,U)\v{e}\v{e}^TC(U,U)}{\v{e}^TC(U,U)\v{e}} \\
    &= C(U,U)^{1/2}\left(I-\dfrac{C(U,U)^{1/2}\v{e}\v{e}^TC(U,U)^{1/2}}{\v{e}^TC(U,U)\v{e}}\right)C(U, U)^{1/2}
    \end{aligned}
\end{equation*}
where the last line from bottom uses $\v{e}^T\v{e} = \lvert S\rvert$ and the last line uses the positive semidefiniteness of $C(U,U)$. Therefore to argue for the positive semidefiniteness of the last term in the above chain, it suffices to show that $I-\dfrac{C(U,U)^{1/2}\v{e}\v{e}^TC(U,U)^{1/2}}{\v{e}^TC(U,U)\v{e}}$ is positive semidefinite. This follows from simply noticing that $\dfrac{C(U,U)^{1/2}\v{e}\v{e}^TC(U,U)^{1/2}}{\v{e}^TC(U,U)\v{e}}$ is a rank-one matrix with eigenvalue $1$. Thus $C$ is positive semidefinite. 

To show that $L-C$ is also positive semidefinite, we show that it is the Laplacian of a graph with non-negative weights. First notice that $C\v{e} = 0$. Indeed, we have
\begin{equation*}
    C\v{e} = \begin{bmatrix}C(U,U)\v{e} + C(U,S)\v{e} \\ C(S,U)\v{e} + C(S,S)\v{e} \\ 0 \end{bmatrix}
\end{equation*}
where
\begin{align*}
    C(U,U)\v{e} + C(U,S)\v{e} &= C(U,U)\v{e} - C(U,U)\v{e} = 0
\end{align*}
using the construction of $C(U,S)$, and
\begin{align*}
    C(S,U)\v{e} + C(S,S)\v{e} &= \dfrac{-\v{e}\v{e}^TC(U,U)\v{e}}{\lvert S\rvert} + \dfrac{\v{e}^TC(U,U)\v{e}\v{e}}{\lvert S\rvert } = 0
\end{align*}
using the constructions of $C(S,U)$ and $C(S,S)$. Subsequently, using the fact that $L\v{e}=0$, we conclude that $(L-C)\v{e} = 0$.

Defining set $R:= V\setminus (U \cup S )$, we now show that each off-diagonal entry of
\begin{equation*}
    L-C = \begin{bmatrix}L(U,U)-C(U,U) & L(U,S)-C(U, S) & L(U, R) \\
    L(S,U)-C(S,U) & L(S,S)-C(S, S) & L(S,R) \\
    L(R,U) & L(R,S) & L(R,R)\end{bmatrix}
\end{equation*}
is non-positive. For each $u, u'\in U$, we have $L_{uu'}-C_{uu'}=0$ by construction. For each $u \in U, s \in S$, we have
\begin{align*}
    L_{us} - C_{us} &= L_{us} + \dfrac{1}{\lvert S\rvert}\sum\limits_{u'\in U}C_{uu'} && (\text{by construction of } C(U,S)) \\
    &= L_{us} + \dfrac{C_{u u}}{\lvert S\rvert} + \dfrac{1}{\lvert S\rvert}\sum\limits_{u'\in U\setminus\{u\}}C_{u u'} \\
    &= L_{us} + \dfrac{1}{\lvert S\rvert} \left( \sum\limits_{u'\in U\setminus\{u\}}\lvert L_{uu'} \rvert + L_{u u'} \right) && (\text{by construction of }C(U,U)) \\
    &= L_{us} + \dfrac{1}{\lvert S\rvert} \sum\limits_{\substack{u'\in U\setminus\{u\}: \\ L_{uu'}>0}} 2L_{u u'}  \\
    &= -W_{us} - \dfrac{1}{\lvert S\rvert} \sum\limits_{\substack{u'\in U\setminus\{u\}: \\ W_{uu'}<0}} 2W_{u u'} &&(\because L=L(G))\\
    &\leq 0. && ({\text{using (\ref{psd2-as})} }) 
\end{align*}

Now it remains to consider the signs of the entries of $L(V, R)$ and the off-diagonal entries of $L(S,S)-C(S,S)$. Observe that each entry of $L(V,R)$ is non-negative since each entry of $W(V, R)$ is non-positive. Indeed any entry in $W(V,R)$ corresponds to an edge whose one endpoint lies in $R$; note that for a negative edge, both endpoints lie in $U$ by definition. Lastly, every off-diagonal entry of $L(S,S)-C(S,S)$ is non-negative since such entries are $0$ in $C(S,S)$, by construction, and non-negative in $L(S,S)$ since they correspond to edges whose both endpoints lie in $S$.

Therefore we have shown that $L-C$ is a Laplacian matrix for a graph with non-negative weights, and is consequently positive semidefinite. Since we have shown the positive semidefiniteness of both $C$ and $L-C$, we conclude that $L$ is positive semidefinite. 
\end{proof}

\subsection{Lack of Robustness}\label{1d-robustness}
As discussed in Section \ref{1d-theory}, we are interested in understanding the robustness of \aoned\ recovery algorithm in the presence of noisy nodes, i.e. fringe nodes that are close to being stray nodes. However, through computational experiments, it is observed that \aoned\ seems to have undesirable behavior in this setting. In particular, there exist pathological instances in which the output of \aoned\ contains groups of noisy nodes as spurious cluster. The following example further illustrates this phenomenon.   

\begin{example} \label{1d-bad1}
Consider a graph $G$ on $n=25$ nodes containing $k=3$ clusters. Suppose $G$ has a cluster $j$ containing $6$ nodes and the three-dimensional features of these nodes are as shown in the rows of the $6 \times 3$ matrix below. 
\begin{equation*}
\begin{bmatrix}
1.00    &      0.00     &     0.00 \\
0.79    &      0.00     &     0.21 \\
1.00    &      0.00     &     0.00 \\
0.53    &      0.47     &     0.00 \\
0.53    &      0.47     &     0.00 \\
0.51    &      0.49     &     0.00
\end{bmatrix}    
\end{equation*}
The submatrix of the output of \aoned\ corresponding to the nodes in cluster $j$ is
\begin{equation*}
\begin{bmatrix}
1 &    1    & 1 &    0 &    0 &    0 \\
1 &    1    & 1 &    0 &    0 &    0 \\
1 &    1    & 1 &    0 &    0 &    0 \\
0 &     0  &   0&     1 &    1 &    1 \\
0 &    0     &0 &    1 &    1 &    1 \\
0 &    0&     0 &    1 &    1 &    1
\end{bmatrix}.
\end{equation*}
The above matrix is not the matrix of all ones. In fact, it breaks down the true cluster into two clusters by creating one cluster each for strong and fringe nodes thereby creating a spurious cluster made up of the fringe nodes. 
\end{example}

This apparent limitation of \aoned\ demotivates a theoretical analysis of cluster recovery using fractional optimal solutions of (\ref{p-1d}).

\section{\texorpdfstring{$\anormd$}{} Recovery Algorithm} \label{nd-sec}

We propose SDP formulation (\ref{p-nd}) obtained by replacing the $n$ diagonal constraints of (\ref{p-1d}) with a single $\ell_2$-norm constraint. Based on this formulation, we propose a novel recovery algorithm, called \anormd, which involves no tuning parameter.

\begin{algorithm}
\caption{\anormd}\label{alg-nd}
\textbf{Input:} Graph $G=(V, W)$ generated according to NFM \\
\textbf{Output:} Symmetric matrix $X^c$ of the same dimension as $W$ whose each entry is in $\{0, 1\}$
\begin{algorithmic}[1]
\STATE $X^* = \arg \max\ \langle W, X\rangle \text{ s.t. } X\geq 0, X\succeq 0,\norm{diag(X)}\leq 1$
\FOR{$t$ in entries of $X^*$ sorted in non-increasing order}
\STATE $X^c_{ij} = $ Round($X^*, t$)
\IF{there exists a non-empty $V' \subseteq V$ such that $X^c(V', V')$ is the cluster matrix for some partition of $V'$}
\STATE \textbf{break}
\ELSE
\STATE $X^c = 0$
\ENDIF
\ENDFOR
\end{algorithmic}
\end{algorithm}
Note that the output of \anormd\ can possibly be the zero matrix and therefore does not define a clustering for the input graph or any of its subgraphs. However, the theory developed in Section \ref{nd-theory} provides conditions on the input graph sufficient for the output of \anormd\ to induce a clustering for some subgraph of the input graph.

\begin{equation}
\label{p-nd}
\tag{P-ND}
\begin{aligned}
& \max_X && \langle W, X \rangle \\
& \: \mathrm{s.t.} && X \geq 0 \\
&&& X \succeq 0 \\
&&& \norm{diag(X)} \leq 1.
\end{aligned}
\end{equation}

SDPs (\ref{p-1d}) and (\ref{p-nd}) both have the non-negativity and positive semidefiniteness constraints on the variable matrix. The combination of these constraints, i.e. the set $\{X : X \geq 0, X \succeq 0\}$, forms the so-called \emph{doubly non-negative (DNN) cone}. This cone has been studied in the context of SDP relaxations for other graph problems such as the minimum cut problem \cite{li2021strictly} and the quadratic assignment problem \cite{hu2019facial, oliveira2018admm}.

The rounding procedure in \anormd\ is based on the observation from computational experiments that the entries of an optimal solution of (\ref{p-nd}) corresponding to the recovered clusters are larger compared to, and therefore well-separated from, the rest of the entries. This implies the existence of a fixed threshold for rounding; however, computational experiments also suggest the dependence of this rounding threshold on problem parameters $n, k$ and $\bm{\alpha}$. An algorithmic dependence on $k$ and $\bm{\alpha}$ is undesirable, especially in Correlation Clustering, since these parameters are latent and encode information about the number and size of clusters in the graph. It is not clear to us whether a fixed-threshold-based rounding procedure exists which does not require prior estimate of $k$ and $\bm{\alpha}$, and this motivates the rounding procedure in \anormd\ which adapts to matrix being rounded. 

\subsection{Theoretical Guarantees}\label{nd-theory}
Similar to our approach for the \aoned\ recovery algorithm, we first prove exact cluster recovery under deterministic conditions on the input graph, followed by understanding the validity of these deterministic conditions for the NFM, and the robustness properties of (\ref{p-nd}).

\begin{as}\label{nd-as}
Suppose $G=(V,W)$ be a graph on $n$ nodes where $W = \v{q} \v{q}^T - D + N$. Let $U\subseteq V$ contain all nodes of $G$ adjacent to a negative edge. That is, $U := \{v \in V: W_{vw} < 0 \text{ for some } w \in V \}$. Suppose the following hold. 
\begin{enumerate}
    \item $\v{q}$ is a positive $n$-dimensional vector satisfying
    \begin{equation*}
        \max(\v{q}^{\circ 4/3}) \leq \dfrac{\norm{\v{q}^{\circ 2/3}}^2 }{45}.
    \end{equation*} \label{v}
Intuitively, this condition is likely to hold if the smallest entry in $\v{q}$ is not too small. 
\item There exists $S \subseteq V \setminus U$ such that for each $u \in U$ and $s \in S$, we have
    \begin{equation*}
        \lvert S\rvert q_s^{1/3}  W_{us} \geq -6 \left( \sum\limits_{\substack{u' \in U: \\ W_{uu'} < 0}} q_{u'}^{1/3}  W_{u u'} \right).
    \end{equation*} \label{w}
\item $N$ is a $n \times n$ symmetric matrix whose each diagonal entry is zero and, for each $i \in [n]$, satisfies
    \begin{equation*}
        \norm{\v{n}^i} \leq \dfrac{q_i \norm{\v{q}^{\circ 2/3}}^2 }{45 \norm{\v{q}^{\circ 1/3}}}.
    \end{equation*}
(Recall the notation that for each $i \in [n]$, $\v{n}^i$ denotes row $i$ of matrix $N$.) \label{n}
\item $D=Diag(\v{q} \circ \v{q})$. This is a \emph{diagonal correction matrix} chosen to ensure that the diagonal of $W$ is indeed zero, as defined. Unlike (\ref{p-1d}), the analysis of (\ref{p-nd}) depends on the diagonal entries of $W$, and therefore it is reasonable to assume each of them to be $0$.  \label{d}
\end{enumerate}
\end{as}

\begin{theorem}\label{nd-rec}
Let $G=(V,W)$ be a graph generated using the NFM, and suppose that for each $j \in [k]$, $W(V_j, V_j)$ satisfies Assumption \ref{nd-as}. Then (\ref{p-nd}) has an optimal solution $X^*$ satisfying $X^*_{ii'} > 0$ if and only if nodes $i$ and $i'$ belong to the same cluster.
\end{theorem}

In terms of proof techniques, unlike the SDP (\ref{p-1d}), (\ref{p-nd}) does not lend itself to an explicit construction of primal-dual optimal solutions and requires a more elaborate argument using Brouwer fixed-point theory. 

Adopting the two-step approach outlined in Section \ref{warmup}, we are now interested in the following two questions:
\begin{itemize}
    \item What is the probability that, for each cluster, the subgraph induced by the union of strong and some fringe nodes satisfies Assumption \ref{nd-as}?
    \item Is the \anormd\ recovery algorithm robust to the presence of noisy nodes, i.e. fringe nodes that are close to being stray?
\end{itemize}

While we do not provide a precise answer to the first question above, we demonstrate, computationally, the connection between Assumption \ref{nd-as} and the NFM in which the distribution over the unit simplex is chosen to be the Dirichlet distribution. For a graph $G=(V,W)$ generated according to NFM, for each $j \in [k]$, define $V_j'$ to be the union of strong nodes and some fringe nodes in cluster $j$ such that the cut-off for selecting fringe nodes depends on problem parameters $k$ and $\bm{\alpha}$. Let $n_j$ be the cardinality of $V_j'$, and let $\Theta_j$ denote the $n_j \times k$ matrix whose rows contain the feature vectors corresponding to the nodes in $V_j'$. Observe that the univariate function $g: (0, 1) \rightarrow \mathbb{R}$ defined as $g(x):= \log(x/(1-x))$ can be approximated using a linear function $l:(0,1) \rightarrow \mathbb{R}$ defined as $l(x) := c\cdot (2x-1)$ for a suitably chosen positive constant $c$. Then for each $j \in [k]$, we have
\begin{equation}\label{w-approx}
    W(V_j', V_j') \approx c\cdot (2\Theta_j \Theta_j^T- E).
\end{equation}
Now through various computational experiments, we notice that the matrix $c\cdot(2\Theta_j \Theta_j^T- E)$ is almost a rank-one matrix such that eigenvector corresponding to the largest eigenvalue is a positive vector. The following example concretely illustrates these observations.  

\begin{example}\label{w-structure}
Consider a graph generated using the NFM with $n=30$ and $k = 3$. The distribution over the unit simplex is chosen to be the Dirichlet distribution with parameter $\bm{\alpha} = 0.3 \v{e}$. For some cluster $j$, let $V_j'$ be the union of strong nodes and fringe nodes whose feature vectors have largest entry at least $0.6$. The three-dimensional features of the nodes in $V_j'$ are shown in the rows of the matrix $\Theta_j$ below.
\begin{equation*}
    \Theta_j = \begin{bmatrix}
          0.05  &        0.83   &       0.11 \\
          0.04  &        0.69    &      0.27 \\
          0.03   &       0.92   &       0.05 \\
          0.02  &        0.73     &     0.25 \\
          0.11   &       0.88   &       0.01 \\
          0.25      &    0.60     &     0.15 \\
          0.00      &    0.99       &   0.01 \\ 
          0.12      &    0.67    &      0.21 \\ 
          0.01      &    0.95     &     0.04 
    \end{bmatrix}
\end{equation*}

We notice that the matrix $c\cdot(2\Theta_j \Theta_j^T- E)$ with $c = 2.2$ has exactly three non-zero eigenvalues given by $8.57$, $0.25$ and $-0.75$. Moreover the unit eigenvector, $\v{v}_j$, corresponding to the eigenvalue $8.57$ is 
\begin{equation*}
    \v{v}_j = \begin{bmatrix}
          0.33 \\
          0.17 \\
          0.43 \\
          0.22 \\
          0.38 \\
          0.06 \\ 
          0.51 \\
          0.15 \\
          0.45 
    \end{bmatrix}
\end{equation*}
which has all positive entries.
\end{example}

The observations made above regarding the spectral properties of the matrix $2.2 \cdot (2\Theta_j \Theta_j^T - E)$ are further shown to be consistent using the results in Table \ref{nfm-w-struct1}, \ref{nfm-w-struct2}, and \ref{nfm-w-struct3}. Each row in these tables corresponds to $10$ cluster instances generated using the NFM in which the simplex distribution is chosen to be the Dirichlet distribution. We fix $k = 3$ and the Dirichlet parameter $\bm{\alpha} = 0.3 \v{e}$. The first column denotes the range corresponding to the size of the subgraph induced by strong nodes and fringe nodes whose feature vectors have largest entry at least $0.6$; the cut-off of $0.6$ is based on manual parameter search for the given setting of $k$ and $\bm{\alpha}$. The second column counts \emph{eigenvector success}, i.e. the number of instances, out of $10$, for which the eigenvector of $2.2 \cdot (2\Theta_j \Theta_j^T - E)$ corresponding to its largest eigenvalue is positive. For such instances, the third column contains the non-zero eigenvalues of $2.2 \cdot (2\Theta_j \Theta_j^T - E)$.

\begin{table}
  \caption{Structure of subgraph induced by strong nodes and some fringe nodes for each cluster (part 1/3).}
  \label{nfm-w-struct1} 
  \centering
  \begin{tabular}{ccc}
    \toprule
    Cluster size range & Eigenvector success $(/10)$ & Non-zero eigenvalues \\
    \midrule
    $6-10$ & $10$ & $7.8, 0.5, -1.1$ \\
    & & $9.8, 0.4, -0.7$ \\
    & & $9.2, 0.2, -0.4$ \\
    & & $10.1, 0.4, -0.3$ \\
    & & $10.1, 0.9, -0.6$ \\
    & & $8.5, 0.1, -0.3$ \\
    & & $12.7, 0.2, -0.6$ \\
    & & $6.6, 0.5, -0.4$ \\
    & & $9.8, 0.4, -0.5$ \\
    & & $7.4, 0.3, -0.4$ \\
    \midrule
    $11-15$ & $10$ & $16.8, 0.2, -0.5$ \\
    & & $14.2, 0.4, -0.4$ \\
    & & $12.9, 0.6, -0.5$ \\
    & & $16.7, 0.6, -0.7$ \\
    & & $15.8, 0.3, -0.5$ \\
    & & $15.9, 0.8, -0.8$ \\
    & & $16.1, 0.2, -0.4$ \\
    & & $14.0, 0.9, -0.8$ \\
    & & $12.8, 0.5, -0.4$ \\
    & & $14.0, 0.3, -0.5$ \\
    \midrule
    $16-20$ & $10$ & $16.3, 1.0, -1.3$ \\
    & & $18.1, 0.7, -0.9$ \\
    & & $17.0, 0.7, -0.7$ \\
    & & $21.5, 1.2, -1.9$ \\
    & & $14.5, 1.2, -1.3$ \\
    & & $21.2, 0.6, -0.7$ \\
    & & $19.4, 1.2, -1.0$ \\
    & & $21.6, 0.3, -0.6$ \\
    & & $23.5, 1.0, -1.0$ \\
    & & $20.1, 0.6, -1.0$ \\
    \bottomrule
  \end{tabular}
\end{table}

\begin{table}
  \caption{Structure of subgraph induced by strong nodes and some fringe nodes for each cluster (part 2/3).}
  \label{nfm-w-struct2} 
  \centering
  \begin{tabular}{ccc}
    \toprule
    Cluster size range & Eigenvector success $(/10)$ & Non-zero eigenvalues \\
    \midrule
    $21-25$ & $10$ & $26.2, 1.7, -1.8$ \\
    & & $21.6, 1.5, -1.4$ \\
    & & $23.9, 1.1, -1.2$ \\
    & & $18.3, 1.6, -1.5$ \\
    & & $26.5, 0.9, -1.2$ \\
    & & $23.0, 2.5, -3.0$ \\
    & & $23.8, 1.6, -1.7$ \\
    & & $25.6, 1.0, -1.0$ \\
    & & $20.0, 2.0, -2.2$ \\
    & & $23.2, 1.4, -1.5$ \\
    \midrule
    $26-30$ & $10$  & $28.6, 1.2, -1.9$ \\
    & & $34.0, 2.0, -2.0$ \\
    & & $32.2, 1.3, -1.6$ \\
    & & $29.0, 2.4, -2.2$ \\
    & & $30.0, 1.6, -1.9$ \\
    & & $30.6, 1.7, -1.8$ \\
    & & $29.9, 1.4, -1.4$ \\
    & & $35.9, 1.3, -1.3$ \\
    & & $23.0, 1.7, -1.4$ \\
    & & $29.1, 1.5, -1.5$ \\
    \midrule
    $31-35$ & $10$ & $38.8, 1.5, -1.5$ \\
    & & $33.7, 1.8, -1.8$ \\
    & & $38.4, 2.3, -2.2$ \\
    & & $31.9, 1.9, -1.7$ \\
    & & $38.3, 1.7, -1.9$ \\
    & & $33.6, 1.9, -2.2$ \\
    & & $36.3, 2.2, -1.9$ \\
    & & $43.7, 1.1, -1.2$ \\
    & & $29.0, 2.8, -3.2$ \\
    & & $46.7, 0.9, -1.5$ \\
    \bottomrule
  \end{tabular}
\end{table}

\begin{table}
  \caption{Structure of subgraph induced by strong nodes and some fringe nodes for each cluster (part 3/3).}
  \label{nfm-w-struct3} 
  \centering
  \begin{tabular}{ccc}
    \toprule
    Cluster size range & Eigenvector success $(/10)$ & Non-zero eigenvalues \\
    \midrule
    $36-40$ & $10$ & $35.6, 1.8, -2.3$ \\
    & & $44.1, 2.0, -2.5$ \\
    & & $40.2, 2.9, -2.8$ \\
    & & $40.2, 2.5, -2.6$ \\
    & & $44.5, 1.2, -2.0$ \\
    & & $49.8, 1.6, -1.5$ \\
    & & $33.5, 2.2, -2.6$ \\
    & & $35.1, 2.1, -2.1$ \\
    & & $46.3, 2.0, -2.1$ \\
    & & $40.1, 1.8, -1.9$ \\
    \bottomrule
  \end{tabular}
\end{table}
These computational results motivate the following conjecture.
\begin{conj}
Let $G=(V,W)$ be a graph generated using the NFM in which the simplex distribution is chosen to be the Dirichlet distribution with constant parameter $\bm{\alpha}$. Then there exists a scalar $t(k, \bm{\alpha}) \in (0.5, 1/\sqrt{2})$ such that for each $j \in [k]$, with probability not converging to $0$ as $n \to \infty$, the largest eigenvalue of the matrix $4.4\Theta_j \Theta_j^T - 2.2E$ is well-separated from the remaining eigenvalues and the corresponding eigenvector is positive where
\begin{equation*}
    V_j' := \{i \in [n]: \theta^i_j \geq t(k, \bm{\alpha}) \}
\end{equation*}
and
\begin{equation*}
\Theta_j := \Theta(V_j', :).
\end{equation*}
\end{conj}

Now let $\v{q}_j := \sqrt{\lambda_j} \v{v}_j$ where $\lambda_j$ and $\v{v}_j$ denote the largest eigenvalue and the corresponding unit eigenvector respectively of $c\cdot (2\Theta_j \Theta_j^T - E)$. We rewrite (\ref{w-approx}) as
\begin{equation}\label{w-approx2}
W(V_j', V_j') = \v{q}_j\v{q}_j^T - D_j + N_j
\end{equation}
where $D_j$ is the $n_j \times n_j$ diagonal matrix $Diag(\v{q}_j\v{q}_j^T)$ and $N_j$ is the $n_j \times n_j$ symmetric matrix whose each diagonal entry is equal to $0$ and the off-diagonal entries are chosen to make so as to make (\ref{w-approx}) hold. That is, matrix $D_j$ applies diagonal correction to ensure $diag(W(V_j, V_j)) = \v{0}$ and matrix $N_j$ captures the error in approximating the logarithmic function $g$ by the linear function $l$ and the error in approximating $c\cdot (2\Theta_j \Theta_j^T - E)$ by $\v{q}_j\v{q}_j^T$. We show, using computational results, the validity of Assumption \ref{nd-as} for quantities $\v{q}_j, N_j, W(V_j, V_j)$ described using (\ref{w-approx2}) in Table \ref{nd-as-expt}. Each row in the table corresponds to $10$ cluster instances generated using the NFM in which the simplex distribution is chosen to be the Dirichlet distribution. We fix $k = 3$ and the Dirichlet parameter $\bm{\alpha} = 0.3 \v{e}$. The first column denotes the range corresponding to the size of the subgraph induced by strong nodes and fringe nodes whose feature vectors have largest entry at least $0.6$; the cut-off of $0.6$ is based on manual parameter search for the given setting of $k$ and $\bm{\alpha}$. The second column counts \emph{C1, C2 success}, i.e. the number of instances, out of $10$, for which vector $\v{q}_j$ and matrix $W(V_j', V_j')$ as highlighted in (\ref{w-approx2}) satisfy conditions \ref{v} and \ref{w} in Assumption \ref{nd-as}. We notice that condition \ref{n} is not satisfied for most instances; however, the violation is by a constant factor in the sense that the ratio 
\begin{equation}\label{c3-ratio}
\dfrac{45 \norm{\v{n}^i} \norm{\v{q}^{\circ 1/3}}} {q_i \norm{\v{q}^{\circ 2/3}}^2}
\end{equation}
for each $i \in [n]$, is bounded above by a constant, albeit much larger than $1$ as desired. Therefore in the fourth column of Table \ref{nd-as-expt}, we present \emph{average C3 upper bound}, i.e. the quantity (\ref{c3-ratio}) first averaged over all nodes in the graph, then averaged over the instances out $10$ runs in which both conditions \ref{v} and \ref{w} are satisfied. 

\begin{table}
  \caption{Verification of Assumption \ref{nd-as}}
  \label{nd-as-expt} 
  \centering
  \begin{tabular}{cccc}
    \toprule
    Cluster size range & C1, C2 success $(/10)$ & Average C3 upper bound \\
    \midrule
    $1201-1300$ & $4$ & $13.3$ \\
    $1301-1400$ & $4$ & $14.1$ \\
    $1401-1500$ & $3$ & $13.2$ \\
    $1501-1600$ & $4$ & $13.3$ \\
    $1601-1700$ & $5$ & $13.4$ \\
    $1701-1800$ & $8$ & $13.5$ \\
    $1801-1900$ & $3$ & $13.6$ \\
    \bottomrule
  \end{tabular}
\end{table}
These computational results partly justify Assumption \ref{nd-as}. Now we turn to the robustness aspect, i.e. understanding the robustness of \anormd\ to the presence of noisy nodes. We begin by revisiting Example \ref{1d-bad1} mentioned in Section \ref{1d-robustness}. In particular, the submatrix of the output of \anormd\ corresponding to the nodes in cluster $j$ is
\begin{equation*}
\begin{bmatrix}
1 &    1    & 1 &    0 &    0 &    0 \\
1 &    1    & 1 &    0 &    0 &    0 \\
1 &    1    & 1 &    0 &    0 &    0 \\
0 &     0  &   0&    0 &    0 &   0 \\
0 &    0     &0 &    0 &    0 &    0 \\
0 &    0&     0 &    0 &    0 &    0
\end{bmatrix}.
\end{equation*}
This shows that \anormd\ is correctly able to cluster the strong nodes, for this example, despite the presence of fringe nodes without creating spurious clusters using the fringe nodes. This observation motivates the following results which show the robustness of the diagonal of an optimal solution of the SDP (\ref{p-nd}) to perturbations of the weighted adjacency matrix $W$. 

\begin{theorem}\label{diag-uniq}
Let $X^*$ be an optimal solution of the SDP (\ref{p-nd}). If $W$ contains at least one positive entry, then $diag(X^*)$ is uniquely determined by $W$. 
\end{theorem}

\begin{theorem}\label{diag-robust}
Let $X^*$ and $X'$ be optimal solutions to the SDP (\ref{p-nd}) for weighted adjacency matrices $W$ and $W + \Delta$ respectively. 
If each of $W$ and $W+\Delta$ contains at least one positive entry, then \begin{equation*}
    \norm{diag(X^*) - diag(X')} \leq \dfrac{2\ (2n)^{1/4}\norm{\Delta}_F^{1/2}}{ \langle W, X^*\rangle^{1/2} }.
\end{equation*}
\end{theorem}

While Theorem \ref{diag-robust} shows the robustness of the diagonal of an optimal solution of (\ref{p-nd}) to only perturbations of the weighted adjacency matrix, it is, in fact, observed using computational experiments that all entries of an optimal solution are robust to the presence of fringe nodes whose feature vectors have a relatively smaller largest entry, i.e. fringe nodes that are close to being stray nodes. In particular, the entries of an optimal solution corresponding to the cluster subgraphs comprised of strong nodes and fringe nodes close to the strong set are larger compared to, and therefore well-separated from, the rest of the entries. This is supported by computational results shown in Table \ref{nd-perf} which show the performance of \anormd. Each row in the table corresponds to $10$ graph instances generated using the NFM in which the simplex distribution is chosen to be the Dirichlet distribution. We fix $k = 3$ and the Dirichlet parameter $\bm{\alpha} = 0.3 \v{e}$. The first column denotes the size of graph, i.e. number of nodes $n$. The second column counts the \emph{\anormd\ success}, i.e. the number of instances, out of $10$, for which the number of recovered clusters is equal to the true number of clusters $k$ such that the recovered clusters are disjoint and each recovered cluster contains all strong nodes (and possibly some fringe nodes) from exactly one ground-truth cluster. 

\begin{table}
  \caption{Performance of \anormd.}
  \label{nd-perf} 
  \centering
  \begin{tabular}{ccc}
    \toprule
    Graph size (number of nodes) & \anormd\ success $(/10)$ \\
    \midrule
    $60$ & $9$ \\
    $70$ & $8$  \\
    $80$ & $10$  \\
    $90$ & $9$  \\
    $100$ & $9$  \\
    $110$ & $10$  \\
    $120$ & $9$  \\
    $130$ & $10$  \\
    $140$ & $10$  \\
    \bottomrule
  \end{tabular}
\end{table}

The theoretical and computational results regarding the performance of \anormd\ presented in this section lead us to make the following conjecture.

\begin{conj}
Let $G$ be a graph generated according to NFM in which the simplex distribution is chosen to be the Dirichlet distribution with constant parameter. Then with probability not converging to $0$ as $n \to \infty$, \anormd\ returns exactly $k$ disjoint clusters $V_1', \dots, V_k'$ such that, for each $j \in [k]$
\begin{equation*}
    V_j^{strong} \subseteq V_j'.
\end{equation*}
\end{conj}

\subsection{Proofs}
In this section we build a proof of Theorem \ref{nd-rec}. Our strategy is to demonstrate the desired structure in each of the submatrices of an optimal solution corresponding to a cluster. One key ingredient for this approach is to determine a point $\v{x}$ such that
\begin{equation*}
(W\v{x})^{\circ 1/3} = \v{x}.    
\end{equation*}
However, note that the exact solution to the system
\begin{equation*}
(\v{q}\v{q}^T \v{x})^{\circ 1/3} = \v{q}    
\end{equation*}
can be shown to be $\v{x} = (\beta \v{q})^{\circ 1/3}$ where $\beta = [\v{q}^T (\v{q}^{\circ 1/3})]^{3/2}$. Moreover, due to Assumption \ref{nd-as}, $W$ can be interpreted as a perturbation of the matrix $\v{q}\v{q}^T$ by matrices $D$ and $N$ thereby motivating the following lemma. 

\begin{lemma}\label{u}
Let $G=(V, W)$ be a graph on $n$ nodes satisfying conditions \ref{v}, \ref{n}, \ref{d} in Assumption \ref{nd-rec}. Then the continuous function $f: \mathbb{R}^n \rightarrow \mathbb{R}^n$ defined as $f(\v{x}) := (W\v{x})^{\circ 1/3}$ maps the set 
\begin{equation*}
    S := \left[\dfrac{1}{2}(\beta \v{q})^{\circ 1/3} , \dfrac{3}{2}(\beta \v{q})^{\circ 1/3} \right], \text{ where } \beta = [\v{q}^T (\v{q}^{\circ 1/3})]^{3/2},
\end{equation*}
to itself. 
\end{lemma}
\begin{proof}
We will show that for any $\v{x} \in S$, each of $f(\v{x}) \geq (\beta \v{q})^{\circ 1/3}/2$ and $f(\v{x}) \leq 3(\beta \v{q})^{\circ 1/3}/2$ holds separately. Pick any $\v{x} \in S$. We have
\begin{equation}\label{main-lb}
    \begin{aligned}
    (W\v{x})^{\circ 1/3} &= (\v{q}\v{q}^T\v{x} - D\v{x} + N\v{x})^{\circ 1/3} \\
    &\geq \left[\dfrac{\beta^{1/3}\v{q}\v{q}^T(\v{q}^{\circ 1/3})}{2} - D\v{x} + N\v{x}\right]^{\circ 1/3} && (\text{using the lower bound on $\v{x}$}) \\
    &= \left[\dfrac{\beta\v{q}}{2} - D\v{x} + N\v{x}\right]^{\circ 1/3} && (\text{using the definition of $\beta$}) \\
    &\geq \left[\dfrac{\beta\v{q}}{2} - \dfrac{3\beta^{1/3}D(\v{q}^{\circ 1/3})}{2} + N\v{x}\right]^{\circ 1/3} && (\text{using the upper bound on $\v{x}$}) \\
    &\geq \left[\dfrac{\beta\v{q}}{2} - \dfrac{3\beta^{1/3}\v{q}^{\circ 7/3}}{2} + N\v{x}\right]^{\circ 1/3}. && (\text{using the definition of $D$}) 
    \end{aligned}
\end{equation}

Now we bound each of the second and third terms above separately. Condition \ref{v} in Assumption \ref{nd-as} implies for each $i \in [n]$,
\begin{align*}
    q_i^{7/3} &\leq \dfrac{\norm{\v{q}^{\circ 2/3}}^2 q_i}{45} \\
    &= \dfrac{\left(\sum\limits_{i \in [n]}q_i^{4/3}\right) q_i}{45} \\
    &= \dfrac{[\v{q}^T (\v{q}^{\circ 1/3})] q_i}{45} \\
    &= \dfrac{\beta^{2/3} q_i}{45}. && (\text{using the definition of $\beta$})
\end{align*}
The above chain implies that 
\begin{equation}\label{second-term}
    \begin{aligned}
    \v{q}^{\circ 7/3} &\leq \dfrac{\beta^{2/3}\v{q}}{45} \\
    \iff \dfrac{3\beta^{1/3}\v{q}^{\circ 7/3}}{2} &\leq \dfrac{\beta \v{q}}{30}. && (\text{multiplying both sides by $3\beta^{1/3}/2$})
    \end{aligned}
\end{equation}
    
Moreover, for each $i \in [n]$, we have
\begin{equation}\label{third-term}
    \begin{aligned}
    \lvert (N\v{x})_i\rvert &= \lvert \v{n}^i \v{x} \rvert \\
    &\leq \norm{ \v{n}^i } \norm{ \v{x}} && (\text{using Cauchy-Schwarz inequality}) \\
    &\leq \dfrac{q_i \norm{\v{q}^{\circ 2/3}}^2}{45 \norm{\v{q}^{\circ 1/3}} } \norm{\v{x}} && (\text{using condition \ref{n} in Assumption \ref{nd-as}}) \\
    &\leq \dfrac{q_i \norm{\v{q}^{\circ 2/3}}^2 \beta^{1/3} }{30} && (\text{using the upper bound on $\v{x}$}) \\
    &= \dfrac{\beta q_i}{30}. && (\text{using the definition of $\beta$})
    \end{aligned}
\end{equation}
Then using (\ref{second-term}) and (\ref{third-term}) in (\ref{main-lb}), we get
\begin{equation}\label{lb}
    \begin{aligned}
    (W\v{x})^{\circ 1/3} &\geq \left(\dfrac{\beta \v{q}}{2} - \dfrac{\beta \v{q}}{15} \right)^{\circ 1/3} \\
    & = \left(\dfrac{13}{30} \right)^{1/3} (\beta \v{q})^{\circ 1/3} \\
    &> \dfrac{1}{2} (\beta \v{q})^{\circ 1/3}.
    \end{aligned}
\end{equation}

For any $\v{x} \in S$, we also have
\begin{equation}\label{ub}
    \begin{aligned}
    (W\v{x})^{\circ 1/3} &= (\v{q}\v{q}^T\v{x} - D\v{x} + N\v{x})^{\circ 1/3} \\
    &\leq \left[\dfrac{3\beta^{1/3}\v{q}\v{q}^T(\v{q}^{\circ 1/3})}{2} - D\v{x} + N\v{x}\right]^{\circ 1/3} && (\text{using the upper bound on $\v{x}$}) \\
    &= \left[\dfrac{3\beta\v{q}}{2} - D\v{x} + N\v{x}\right]^{1/3} && (\text{using the definition of $\beta$}) \\
    &\leq \left[\dfrac{3\beta\v{q}}{2} + N\v{x}\right]^{\circ 1/3} && (\because D\v{x} > 0) \\
    &\leq \left(\dfrac{23}{15}\right)^{1/3} (\beta \v{q})^{\circ 1/3} && (\text{using (\ref{third-term})}) \\
    &< \dfrac{3}{2} (\beta \v{q})^{\circ 1/3}.
    \end{aligned}
\end{equation}
Combining (\ref{lb}) and (\ref{ub}), we conclude that the function $f$ maps $S$ to itself. 
\end{proof}

To proceed with the proof of Theorem \ref{nd-rec}, in addition to Lemma \ref{u}, we also make use of the following Brouwer fixed-point theorem.

\begin{theorem}[Brouwer Fixed-Point Theorem \cite{brouwer1912}]\label{brouwer-fp}
Let $C \subseteq \mathbb{R}^n$ be a non-empty convex compact set and let $f: C \rightarrow C$ be a continuous function. Then there exists a point $\v{x} \in C$ such that $f(\v{x}) = \v{x}$.
\end{theorem}

\begin{proof}[Proof of Theorem \ref{nd-rec}]
For each $j \in [k]$, let $\v{q}_j, D_j, N_j, U_j$ and $S_j$ denote the quantities mentioned in Assumption \ref{nd-as}, and define $n_j$ as the cardinality of $V_j$ (i.e. the size of cluster $j$). Our analysis uses SDP duality and therefore note that the dual of (\ref{p-nd}) is
\begin{equation}
\label{d-nd}
\tag{D-ND}
\begin{aligned}
& \min_{(X, Y, Z, \lambda)} && \lambda \norm{diag(X)}^2 + \lambda \\
& \: \mathrm{s.t.} && Y \geq 0 \\
&&& Z \succeq 0 \\
&&& \lambda \geq 0 \\
&&& W+Y+Z = \lambda \cdot Diag(X).
\end{aligned}
\end{equation}
Both optimization problems (\ref{p-nd}) and (\ref{d-nd}) have strictly feasible solutions. For instance, $X' := (0.5I + 0.5E)/n$ is a positive, positive definite matrix which is feasible for (\ref{p-nd}). Similarly, $X' := I, Y' := E$, $Z' := (\norm{W+E} + \epsilon) I - (W+E) $ and $\lambda' := (\norm{W+E} + \epsilon)$ gives a strictly feasible solution $(X', Y', Z', \lambda')$ for (\ref{d-nd}) for any $\epsilon > 0$. Therefore using the Karush-Kuhn-Tucker (KKT) conditions for optimality, $X^*$ is an optimal solution for (\ref{p-nd}) if and only if $X^*$ is feasible for (\ref{p-nd}) and there exist a non-negative matrix $Y^* \in \mathbb{S}^n$, a positive semidefinite matrix $Z^*$, and a non-negative scalar $\lambda^*$ such that:
\begin{itemize}
    \item $X^*_{ij} Y^*_{ij} = 0, \forall i, j \in [n]$ 
    \item $\langle X^*, Z^*\rangle = 0$
    \item $\lambda^* \cdot (\norm{diag(X^*)}-1) = 0$
    \item $W+Y^*+Z^* = \lambda^* \cdot Diag(X^*)$
\end{itemize}
In the remainder of the proof, we will explicitly construct all the above mentioned quantities. Now for each $j \in [k]$, since $W(V_j, V_j)$ satisfies Assumption \ref{nd-as}, using Lemma \ref{u} and Theorem \ref{brouwer-fp}, we conclude that there exists an $n_j$-dimensional vector $\v{r}_j \in \left[\dfrac{1}{2}(\beta_j \v{q}_j)^{\circ 1/3}, \dfrac{3}{2}(\beta_j \v{q}_j)^{\circ 1/3} \right]$, where $\beta_j = [\v{q}_j^T (\v{q}_j^{\circ 1/3})]^{3/2}$, such that
\begin{equation}\label{fp-eq}
    W(V_j, V_j)\cdot \v{r}_j = \v{r}_j^{\circ 3}.
\end{equation}

We set
\begin{equation*}
\lambda^* = \sqrt{\sum\limits_{j \in [k]} \norm{\v{r}_j \circ \v{r}_j}^2}.
\end{equation*}

For each $j \in [k]$, we set
\begin{align*}
X^*(V_j, V_j) &=\v{r}_j \v{r}_j^T/\lambda^* \\ 
Y^*(V_j, V_j) &= 0 \\
Z^*(V_j, V_j) &= Diag(\v{r}_j \circ \v{r}_j) - W(V_j, V_j).
\end{align*}

For each distinct $j, j' \in [k]$, we set
\begin{align*}
X^*(V_j, V_{j'}) &= 0 \\
Y^*(V_j, V_{j'}) &= -W(V_j, V_{j'}) \\
Z^*(V_j, V_j) &= 0.
\end{align*}

For each stray node $v$, we set
\begin{align*}
X^*(v, :) &= 0 && \text{(and  $X^*(:, v) = 0$)} \\
Y^*(v, :) &= -W(v, :) && \text{(and  $Y^*(:, v) = -W(:, v)$)} \\
Z^*(v, :) &= 0. && \text{(and  $Z^*(:, v) = 0$)}
\end{align*}

First we show that the constructed $X^*$ is feasible for (\ref{p-nd}). Note that there exists a permutation of the rows (and columns) of $X^*$ which yields a block diagonal matrix in which the non-zero blocks are given by the rank-one positive semidefinite matrices $\v{r}_1\v{r}_1^T/\lambda^*, \dots, \v{r}_k\v{r}_k^T/\lambda^*$. This shows that $X^*$ is positive semidefinite. Morever, since vectors $\v{r}_1, \dots, \v{r}_k$ are positive, we conclude that $X^*$ is non-negative. We also have
\begin{align*}
\norm{diag(X^*)}^2 &= \dfrac{\sum\limits_{j \in [k]} \norm{\v{r}_j \circ \v{r}_j}^2}{{\lambda^*}^2} \\
&= 1. && (\text{using the definition of $\lambda^*$})
\end{align*}
Therefore $X^*$ is feasible for (\ref{p-nd}). Note that this also implies that $\lambda^*(\norm{diag(X^*)}-1) = 0$.

Now we show that the constructed $Y^*, Z^*, \lambda^*$ satisfy the remaining desired properties. Because each pair of nodes lying in distinct clusters shares a negative edge and because each stray node shares a negative edge with every other node in the graph, we have that $Y^* \geq 0$. Also note that $\lambda^*$ is a positive scalar by construction. 

Matrices $X^*$ and $Y^*$ have disjoint supports by construction, and therefore $X^*_{ij}Y^*_{ij} = 0$ for each $i, j\in [n]$. Moreover, for each $j \in [k]$, using (\ref{fp-eq}), we have
\begin{equation}\label{xz}
    \begin{aligned}
    & W(V_j, V_j) \cdot \v{r}_j = \v{r}_j^{\circ 3} \\
    \iff & W(V_j, V_j) \cdot \v{r}_j = Diag(\v{r}_j \circ \v{r}_j) \cdot \v{r}_j \\
    \iff & Z^*(V_j, V_j) \cdot \v{r}_j =  0  && (\text{using the definition of $Z^*$}) \\
    \iff & Z^*(V_j, V_j) \cdot \v{r}_j \v{r}_j^T/ \lambda^* =  0  && (\because \v{r}_j > 0, \lambda^* > 0) \\
    \iff & Z^*(V_j, V_j) \cdot X^*(V_j, V_j) =  0  && (\text{using the definition of $X^*$}) \\
    \iff & \langle Z^*(V_j, V_j) , X^*(V_j, V_j) \rangle =  0  && (\because X^*, Z^* \succeq 0) 
    \end{aligned}
\end{equation}
Therefore we have
\begin{align*}
    \langle X^*, Z^* \rangle &= \sum\limits_{j \in [k]} \langle X^*(V_j, V_j), Z^*(V_j, V_j) \rangle && (\text{using the definitions of $X^*, Z^*$}) \\
    &= 0. && (\text{using (\ref{xz})})
\end{align*}

Note that there exists a permutation of the rows (and columns) of $Z^*$ which yields a block diagonal matrix in which the non-zero blocks are given by the matrices $Diag(\v{r}_1 \circ \v{r}_1) - W(V_1, V_1), \dots, Diag(\v{r}_k \circ \v{r}_k) - W(V_k, V_k)$. Therefore to show the positive semidefiniteness of $Z^*$, it suffices to show that for each $j \in [k]$, the matrix $Z^*(V_j, V_j) = Diag(\v{r}_j \circ \v{r}_j) - W(V_j, V_j)$ is positive semidefinite. From (\ref{xz}), we know that $Z^*(V_j, V_j) \cdot \v{r}_j = 0$. This implies that $\v{e}$ belongs to the null space of $Diag(\v{r}_j) \cdot Z^*(V_j, V_j) \cdot Diag(\v{r}_j)$. Consequently, we observe that
\begin{equation*}
 \bar{L}_j := Diag(\v{r}_j) \cdot Z^*(V_j, V_j) \cdot Diag(\v{r}_j)   
\end{equation*}
is the Laplacian matrix of a graph, called $\bar{G}_j$, on $n_j$ nodes whose weighted adjacency matrix is
\begin{equation*}
\bar{W}_j := Diag(\v{r}_j) \cdot W(V_j, V_j) \cdot Diag(\v{r}_j).
\end{equation*}
Moreover, $Z^*$ is positive semidefinite if and only if the Laplacian $\bar{L}_j$ is positive semidefinite since each entry of $\v{r}_j$ is positive. Note that the sign of each edge in $\bar{G}_j$ is identical to that of the corresponding edge in $G[V_j]$ which implies that the set of all nodes in $\bar{G}_j$ adjacent to a negative edge is $U_j$. Now for any $u \in U_j$ and $s \in S_j$, we have
\begin{align*}
    \lvert S_j \rvert \bar{W}_j(u, s) &= \lvert S_j \rvert r_j(u) r_j(s) W(u,s) \\
    &\geq \lvert S_j \rvert r_j(u) \dfrac{[\beta_j q_j(s)]^{1/3}}{2} W(u,s) \\
    &\hspace{2cm} (\text{using the lower bound on $\v{r}_j$})\\
    & \geq -3 \beta_j^{1/3} \left( \sum\limits_{\substack{u' \in U_j: \\ W(u,u')< 0}} r_j(u) q_j(u')^{1/3}  W(u, u') \right) \\
    &\hspace{2cm} (\text{using condition \ref{w} in Assumption \ref{nd-as}}) \\
    & \geq 2 \left( \sum\limits_{\substack{u' \in U_j: \\ W(u,u')< 0}} r_j(u) r_j(u')  W(u, u') \right) \\
    & \hspace{2cm} (\text{using the upper bound on $\v{r}_j$}) \\
    &= 2 \left( \sum\limits_{\substack{u' \in U_j: \\ \bar{W}_j(u,u')< 0}} \bar{W}_j(u, u') \right). \\
    &\hspace{2cm} (\text{using the definition of $\bar{W}$})
\end{align*}
Thus we have shown that graph $\bar{G}_j$ satisfies (\ref{psd2-as}) stated in Theorem \ref{psd2} using which we conclude that $\bar{L}_j$ is positive semidefinite.

Lastly, we show that the equation $W + Y^* + Z^* = \lambda^* \cdot Diag(X^*)$ is satisfied. For each $j \in [k]$, we have
\begin{align*}
    W(V_j, V_j) + Y^*(V_j, V_j) + Z^*(V_j, V_j) &= Diag(\v{r}_j \circ \v{r}_j) \\
    &\hspace{2cm} (\text{using the definitions of $Y^*, Z^*$}) \\
    &= \lambda^* \cdot Diag(X^*(V_j, V_j)). \\
    &\hspace{2cm} (\text{using the definition of $X^*$})
\end{align*}
For each distinct $j, j' \in [k]$, we have
\begin{align*}
    W(V_j, V_{j'}) + Y^*(V_j, V_{j'}) + Z^*(V_j, V_{j'}) &= 0
\end{align*}
using the definitions of $Y^*, Z^*$. Similarly, for each stray node $v$, we have
\begin{align*}
    W(v, :) + Y^*(v, :) + Z^*(v, :) &= 0 \\
    W(:, v) + Y^*(:, v) + Z^*(:, v) &= 0 
\end{align*}
using the definitions of $Y^*, Z^*$.
\end{proof}

\begin{proof}[Proof of Theorem \ref{diag-uniq}]
Observe that $X=0$ is a feasible solution for (\ref{p-nd}) which implies that the optimal value of (\ref{p-nd}) is non-negative. This implies that for any optimal solution, without loss of generality, we may assume that the constraint $\norm{diag(X)} \leq 1$ is tight. Since $X^*$ is optimal for (\ref{p-nd}), and since both (\ref{p-nd}) and its dual have strictly feasible solutions, using the Karush-Kuhn-Tucker (KKT) conditions for optimality, there exist a non-negative matrix $Y^* \in \mathbb{S}^n$, a positive semidefinite matrix $Z^*$, and a non-negative scalar $\lambda^*$ such that:
\begin{itemize}
    \item $X^*_{ij} Y^*_{ij} = 0, \forall i, j \in [n]$ 
    \item $\langle X^*, Z^*\rangle = 0$
    \item $\lambda^* \cdot (\norm{diag(X^*)}-1) = 0$
    \item $W+Y^*+Z^* = \lambda^* \cdot Diag(X^*)$ 
\end{itemize}

We also note that $\lambda^*$ is a positive scalar. Indeed if $\lambda^*$ is zero, then the last condition above implies $diag(Z^*)$ is zero and consequently $Z^*=0$ since $Z^*$ is positive semidefinite. This implies that $Y^*= -W$ which contradicts the non-negativity of $Y^*$ since $W$ contains a positive entry. 

Let $X^{**}$ be another optimal solution of (\ref{p-nd}). Then we have
\begin{align*}
    0 &= \langle W, X^*-X^{**}\rangle \\ 
    &= \langle \lambda^* \cdot Diag(X^*) - Y^*- Z^*, X^*-X^{**} \rangle && (\text{substituting for $W$}) \\
    &= \lambda^*  - \lambda^* \cdot \langle Diag(X^*), X^{**} \rangle - \langle Y^*+Z^*, X^*-X^{**}\rangle && (\because \norm{diag(X^*)} = 1) \\
    &= \lambda^*  - \lambda^* \cdot \langle Diag(X^*), X^{**} \rangle + \langle Y^*+Z^*, X^{**}\rangle && (\because \langle Y^*, X^*\rangle = \langle Z^*, X^*\rangle = 0) \\
    &\geq \lambda^*  - \lambda^* \cdot \langle Diag(X^*), X^{**} \rangle. && (\because \langle Y^*, X^{**}\rangle, \langle Z^*, X^{**}\rangle \geq 0)
\end{align*}
Using the fact that $\lambda^*$ is positive, the above implies that $\langle diag(X^*), diag(X^{**})\rangle \geq 1$. However, since both $diag(X^*)$ and $diag(X^{**})$ lie on the unit sphere, we conclude that $diag(X^*) = diag(X^{**})$.
\end{proof}

\begin{proof}[Proof of Theorem \ref{diag-robust}]
Note that since $W$ has at least one positive entry, $\max(W_+)$ is a positive scalar. If $W_{ii'} > 0$ for some $i, i' \in [n]$, then $(\v{e}_i \v{e}_{i'}^T + \v{e}_{i'} \v{e}_i^T)/\sqrt{2}$ is feasible for (\ref{p-nd}) and we have
\begin{equation}
    \langle W, X^* \rangle \geq \sqrt{2} W_{ii'} > 0.
\end{equation}
Using a similar argument, we also conclude that
\begin{equation}
    \langle W + \Delta, X' \rangle > 0.
\end{equation}
Moreover since the optimal values of the two programs are positive, we have that
\begin{equation*}
\norm{diag(X^*)} = \norm{diag(X')} = 1.     
\end{equation*}
Observe that
\begin{equation}\label{xs}
    \begin{aligned}
    \lvert \langle \Delta, X^* \rangle \rvert &\leq \norm{\Delta}_F \norm{X^*}_F && (\text{using Cauchy-Schwarz inequality}) \\
    &\leq \sqrt{n}\norm{\Delta}_F \norm{diag(X^*)} && (\because \norm{X^*}_F \leq \sqrt{n}\norm{diag(X^*)}) \\
    &= \sqrt{n}\norm{\Delta}_F. && (\because \norm{diag(X^*)} = 1)
    \end{aligned}
\end{equation}
Similarly, we also have that
\begin{equation} \label{xp}
    \lvert \langle \Delta, X' \rangle \rvert \leq \sqrt{n}\norm{\Delta}_F.
\end{equation}
Now define 
\begin{equation*}
    X'' := \dfrac{X^*+X'}{\norm{diag(X^*) + diag(X')}}.
\end{equation*}
Noting that $X''$ is feasible for (\ref{p-nd}), and therefore using the fact that $\langle W, X^*\rangle \geq \langle W, X'' \rangle $, we get
\begin{align*}
    \langle W, X^*\rangle \norm{diag(X^*) + diag(X')} & \geq \langle W,X^* \rangle + \langle W, X' \rangle \\
    &= \langle W,X^* \rangle + \langle W + \Delta, X' \rangle - \langle \Delta, X' \rangle \\
    &\geq \langle W,X^* \rangle + \langle W + \Delta, X^* \rangle - \langle \Delta, X' \rangle \\
    & \hspace{2cm} (\text{using the optimality of $X'$}) \\
    &= 2\langle W,X^* \rangle + \langle \Delta, X^* \rangle - \langle \Delta, X' \rangle \\
    &\geq 2\langle W,X^* \rangle - 2\sqrt{n}\norm{\Delta}_F \\
    & \hspace{2cm} (\text{using (\ref{xs}) and (\ref{xp})}) 
\end{align*}
which is equivalent to 
\begin{equation}\label{diag-norm-bd}
\norm{diag(X^*) + diag(X')} \geq 2 - \dfrac{2\sqrt{n}\norm{\Delta}_F}{\langle W, X^*\rangle}
\end{equation}
since $\langle W, X^*\rangle$ is positive. Now we have
\begin{align*}
    \norm{diag(X^*) - diag(X')} &= \sqrt{4 - \norm{diag(X^*) + diag(X')}^2} \\
    & \hspace{2cm} (\because \norm{diag(X^*)} = \norm{diag(X')} = 1) \\
    &\leq \dfrac{2\sqrt{2}n^{1/4}\norm{\Delta}_F^{1/2}}{\langle W, X^*\rangle^{1/2}}. \\
    &\hspace{2cm} (\text{using (\ref{diag-norm-bd})})
\end{align*}
This concludes the proof.  
\end{proof}

\section{Conclusions}

In this work, we propose a novel generative model, NFM, for graphs which, unlike the SBM, also generates feature vectors for each node in the graph. We analyze, theoretically and computationally, the performance of two different SDP formulations in recovering the true clusters in graph instances generated according to the NFM. In particular, we begin with an algorithm based on the SDP (\ref{p-1d}), but then demonstrate its lack of robustness to certain noisy instances generated by the NFM. To overcome this shortcoming, we propose a new algorithm based on a different SDP (\ref{p-nd}). We build theory towards showing that SDP (\ref{p-nd}) can be used to provably recover, for each true cluster, nodes with sufficiently strong membership signal in their feature vectors, in the presence of noisy nodes, without involving any tuning parameters.

\bibliography{robust_cc}

\end{document}